\documentclass[a4paper,11pt]{article}

\usepackage[metapost, truebbox]{mfpic}
\usepackage{epsf}
\usepackage{enumitem}

\usepackage{enumitem}

\usepackage{graphicx}
\usepackage{float}

\usepackage{url}

\usepackage{tikz}
\usetikzlibrary{petri}

\usepackage{amsmath}
\usepackage{amssymb}
\usepackage{bbold}
\usepackage{amsthm}

\usepackage{multirow}

\def\cR{{\mathcal R}}

\def\cS{{\mathcal S}}

\def\cC{{\mathcal C}}

\def\cN{{\mathcal N}}
\def\cM{{\mathcal M}}

\def\cA{{\mathcal A}}
\def\w{{\omega}}

\def\eref#1{(\ref{#1})}

\newcommand{\piun}{\ensuremath{{\sf \Pi}}}
\newcommand{\pideux}{\ensuremath{{\sf \Pi}^2}}
\newcommand{\pitrois}{\ensuremath{{\sf \Pi}^3}}
\newcommand{\setN}{\ensuremath{\mathbb{N}}}
\newcommand{\setZ}{\ensuremath{\mathbb{Z}}}
\newcommand{\Q}{\ensuremath{\mathbb{Q}}}

\opengraphsfile{petri}

\numberwithin{equation}{section}
\newtheorem{theo}{Theorem}[section]
\newtheorem{prop}[theo]{Proposition}
\newtheorem{defi}[theo]{Definition}
\newtheorem{coro}[theo]{Corollary}
\newtheorem{lemm}[theo]{Lemma}

\title{Synthesis and Analysis of Product-form Petri Nets}

\author{Serge Haddad\\ \small{ENS Cachan, LSV},\\ \small{CNRS UMR 8643, INRIA,
Cachan, France}\\
\small{\tt haddad{@}lsv.ens-cachan.fr}\\
\and Jean Mairesse, Hoang-Thach Nguyen\\ \small{Universit\'e Paris 7, LIAFA,}\\
\small{CNRS UMR 7089, Paris,
France}\\
\small{\tt \{mairesse, ngthach\}{@}liafa.jussieu.fr}}
\date{}

\begin{document}
\maketitle

\begin{abstract}
For a large Markovian model, a ``product form'' is an
explicit description of the steady-state behaviour which is
otherwise generally untractable. Being first introduced in queueing networks,
it has been adapted to Markovian Petri nets. Here we address three relevant
issues for product-form Petri nets which were left fully or partially open: (1) we provide
a sound and complete set of rules for the synthesis; (2) we characterise
the exact complexity of classical problems like reachability; (3) we introduce a new subclass
for which the normalising constant (a crucial value for product-form expression) can be efficiently
computed.


\end{abstract}


\section{Introduction}
\label{sec:introduction}

{\bf Product-form for stochastic models.}
Markovian models of discrete events systems are
powerful formalisms for modelling and evaluating the
performances of such systems. The main goal is the equilibrium 
performance analysis. It requires to
compute the stationary distribution of a continuous time Markov
process derived from the model.
Unfortunately 
the potentially huge (sometimes infinite) state space of the models often
prevents the modeller from computing explicitly this distribution. 
To cope with the issue, one can forget about exact solutions and settle
for approximations, bounds, or even simulations. The other possibility
is to focus on subclasses for which some kind of explicit description
is indeed possible. In this direction, the most efficient and
satisfactory approach may be the \emph{product-form}
method: for a model composed of modules, the stationary probability of
a  global state may be expressed as a product of  quantities depending
only on local states divided by a \emph{normalising constant}.

 Such a method is applicable when the interactions
between the modules are ``weak''. This is the case 
for queueing networks where the interactions between queues
are described by a random routing of clients. Various
classes of  queueing networks with product-form solutions
have been exhibited~\cite{jack63,bcmp,kell79}.
Moreover efficient algorithms have been designed for
the computation of the normalising constant~\cite{reis80}.

\noindent
{\bf Product-form Petri nets.}
Due to the explicit modelling of competition and synchronisation, the
Markovian Petri nets formalism~\cite{gspnbook}
is an attractive modelling
paradigm.
Similarly to queueing networks, product-form Markovian Petri Nets
were introduced to cope with the combinatorial explosion
of the state space. Historically, works started with purely
behavioural properties 
(i.e. by an analysis of the reachability graph) as in~\cite{lazar-robert87},
and then progressively moved to more and more structural
characterisations \cite{li_geor91,heluta-jams}.
Building on the work of~\cite{heluta-jams}, the authors of~\cite{Haddad05}
establish the first purely structural condition for which 
a product form exists and propose a polynomial time algorithm to check
for the condition, see also \cite{MairesseNguyen09} for an alternative
characterisation. These nets are called $\pideux$-nets.

Product-form Petri nets have been applied for the specification and analysis
of complex systems. From a modelling point of view, compositional
approaches have been proposed~\cite{Balsamo2011,BM2009} as well
as hierarchical ones~\cite{Harrison2011}. Application fields
have also been identified like (1) hardware design and more particularly
RAID storage~\cite{Harrison2011} and (2) software architectures~\cite{BalsamoMarin2011}.

\noindent
{\bf Open issues related to product-form Petri nets.}
\begin{itemize}[nolistsep]
	\item From a modelling point of view, it is more interesting to
          design specific types of
	Petri nets by modular constructions rather than checking a posteriori
	whether a net satisfies the specification. For instance,
	in~\cite{esparza94}, a sound and complete set of rules is proposed for
	the synthesis of live and bounded free-choice nets. Is it possible
        to get an analog for product-form Petri nets?
	\item From a qualitative analysis point of view, it is
          interesting to know the complexity 
	of classical problems (reachability,
        coverability, liveness, etc.) for a given subclass of Petri nets
        and to compare it with that of general Petri nets. For
        product-form Petri nets, partial results were presented
        in~\cite{Haddad05} but several questions were left open. For
        instance, the reachability problem is {\sf PSPACE}-complete for
        safe Petri nets 
	but in safe product-form Petri nets it is only proved to be
        {\sf NP}-hard in~\cite{Haddad05}.  

       \item From a quantitative analysis point of view, an important
         and difficult
         issue is the computation of the
         normalising constant. Indeed, in product-form Petri nets, 
         one can directly compute relative probabilities
         (e.g. available versus unavailable 
       service), but 
       determining absolute probabilities requires to compute the normalising
       constant (i.e. the sum over
       reachable states of the relative probabilities). 
       In models of queueing networks, this can be efficiently performed using
       dynamic programming. In Petri nets, it has been proved that the
       efficient computation 
       is possible when the linear invariants characterise
       the set of reachable markings~\cite{coheta}.  Unfortunately, all
       the known subclasses of product-form nets 
       that fulfill this characterisation are models of queueing
       networks! 
\end{itemize}

 \noindent
{\bf Our contribution.} Here we address the three above issues. In
Section \ref{section:synthesis}, we provide
a set of sound and complete rules for generating any \pideux-net.
We also 
use these rules for transforming a
general Petri net into a related
product-form Petri net. In Section \ref{section:qualitative}, we solve relevant
complexity issues. More precisely, we show that
the reachability and liveness problems are {\sf PSPACE}-complete for
safe product-form 
nets and that the coverability problem is  {\sf EXPSPACE}-complete for
general product-form 
nets. From these complexity results, we conjecture that the problem of computing
the normalising constant does not admit an efficient solution for the
general class of 
product-form Petri nets. However, in Section
\ref{section:normalization}, we introduce a large subclass of
product-form 
Petri nets, denoted $\pitrois$-nets, for which the normalising
constant can be efficiently 
computed. We emphasise that contrary to all subclasses related to
queueing networks,  
$\pitrois$-nets may admit \emph{spurious} markings (i.e. that fufill
the invariants 
while being unreachable).

The above results may change our perspective on product-form Petri
nets. It is proved in \cite{MairesseNguyen09} that the intersection of
free-choice and product-form Petri nets is the class of
Jackson networks~\cite{jack63}. This may suggest that the class of
product-form Petri nets is somehow included in the class of
product-form queueing networks. In the present paper, we
refute this belief in two ways. First by showing that some classical
problems are as complex for product-form Petri nets as for general
Petri nets whereas they become very simple for product-form queueing
networks. Second by exhibiting the large class of $\pitrois$-nets
which can model complex behaviours (e.g. illustrated by the presence
of spurious markings). 

\smallskip

A conference version of the paper appeared in \cite{HMNg11}.
The present version includes additional results (Subsection
\ref{subsec:liveboundedpi2}) together with full proofs of the
results. (There is one exception, Proposition \ref{prop:coverability}, for which the
proof can be found in the arXiv version of the 
paper available at~\url{http://arxiv.org/abs/1104.0291}) 

\smallskip

 \noindent
{\bf Notations.} We often denote a vector $u \in \mathbb R^S$ by $\sum_s u(s)
s$. The \emph{support} of vector $u$ is the subset $S'\equiv \{s \in S \mid u(s)\neq 0\}$.

\section{Petri nets, product-form nets, and \pideux-nets}
\label{section:definition}

\begin{defi}[Petri net]
  \label{def:petrinets}
  A {\em Petri net} is a 5-tuple $\cN= (P,T,W^-,W^+,m_0)$ where:
\begin{itemize}[nolistsep]
  \item $P$ is a finite set of {\em places};
  \item $T$ is a finite set of {\em transitions}, disjoint from $P$;
  \item $W^-$, resp. $W^+$, is a $P \times T$ matrix with coefficients
    in $\setN$; 
  \item $m_0 \in \setN^P$ is the {\em initial marking}.
\end{itemize}
\end{defi}

Below, we also call {\em Petri net} the unmarked quadruple
$(P,T,W^-,W^+)$. The presence or absence of a marking will depend on
the context.

A Petri net is represented in Figure \ref{fig:egPN}. The following
graphical conventions are used: places are represented 
by circles  and transitions by rectangles. There is an arc from $p\in
P$ to $t\in T$ (resp. from $t\in T$ to $p\in P$) if
$W^+(p,t)>0$ (resp. $W^-(p,t)>0$), and the weight
$W^+(p,t)$ (resp. $W^-(p,t)$) is written above the corresponding
arc except when it is equal to 1 in which case it is omitted. 
The initial marking is 
materialised: if $m_0(p) = k$, then $k$ tokens are drawn inside the
circle $p$.  Let $P' \subset P$ and $m$ be a marking then $m(P')$
is defined by $m(P')\equiv \sum_{p \in P'} m(p)$.

The matrix $W=W^+-W^-$ is the {\em incidence
matrix} of the Petri net. 
The {\em input bag} $^\bullet{t}$ (resp.
 {\em output bag} $t^\bullet$) of the transition $t$ is the column
 vector of $W^-$ (resp. $W^+$) 
indexed by $t$. For a place $p$, we define $^\bullet{p}$ and
$p^\bullet$ similarly. 
A {\em T-semi-flow} (resp. {\em S-semi-flow}) is a $\mathbb Q$-valued vector
$v$ such that $W.v = (0,\dots, 0)$ (resp. $v.W = (0,\dots , 0)$).

A {\em symmetric} Petri net is a Petri net such that: $\forall t \in
T, \ \exists t^- \in T, \quad 
^\bullet{t} = (t^-)^\bullet, t^\bullet = {}^\bullet{t^-}$. 
A {\em free-choice net} is a Petri net such that: $\forall t, t' \in T$, either
${}^{\bullet}t \cap {}^{\bullet}t' = \emptyset $, or ${}^{\bullet}t =
{}^{\bullet}t'$.
A {\em state machine} is a Petri net such that:  $\forall t \in
T, \ |^\bullet{t}| = |t^\bullet| =1$. 
A {\em marked graph} is a Petri net such that: $\forall p \in P, \ |^\bullet{p}|
= |p^\bullet| =1$. 

\begin{defi}[Firing rule]
A transition $t$ is {\em enabled} by the marking $m$ if $m \geq
{^\bullet t}$ (denoted by $m
\stackrel{t}{\longrightarrow}$); an
enabled transition $t$ may {\em fire} which transforms the marking $m$
into $m - {^\bullet t} + t^\bullet$, denoted by
$m
\stackrel{t}{\longrightarrow} m' = m 
- {^\bullet t} + t^\bullet$.
\end{defi}


\begin{figure}[htbp]
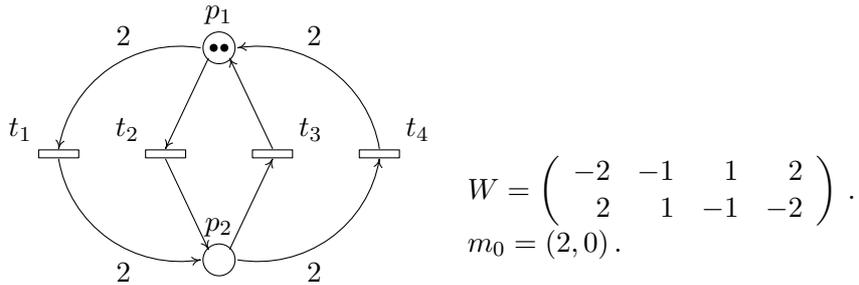

\centering
\begin{mfpic}{-140}{140}{-50}{50}
\tlabelsep{3pt}

\circle{(-70, 40), 6}
\circle{(-70, -40), 6}
\tlabel[bc](-70, 46){$p_1$}
\tlabel[bc](-70, -34){$p_2$}

\shiftpath{(-130, 0)}\rect{(-7.5, -1.5), (7.5, 1.5)}
\shiftpath{(-90, 0)}\rect{(-7.5, -1.5), (7.5, 1.5)}
\shiftpath{(-50, 0)}\rect{(-7.5, -1.5), (7.5, 1.5)}
\shiftpath{(-10, 0)}\rect{(-7.5, -1.5), (7.5, 1.5)}
\tlabel[br](-137, 2){$t_1$}
\tlabel[br](-97, 2){$t_2$}
\tlabel[bl](-43, 2){$t_3$}
\tlabel[bl](-3, 2){$t_4$}

\arrow\polyline{(-74, 36), (-90, 2)}
\arrow\polyline{(-90, -2), (-74, -36)}
\arrow\polyline{(-66, -36), (-50, -2)}
\arrow\polyline{(-50, 2), (-66, 36)}

\arrow\arc[s]{(-77, 40), (-130, 2), 90}
\arrow\arc[s]{(-130, -2), (-77, -40), 90}
\arrow\arc[s]{(-63, -40), (-10, -2), 90}
\arrow\arc[s]{(-10, 2), (-63, 40), 90}
\tlabel[br](-100, 38){$2$}
\tlabel[tr](-100, -38){$2$}
\tlabel[tl](-40, -38){$2$}
\tlabel[bl](-40, 38){$2$}

\point[3pt]{(-72, 40), (-68, 40)}

\tlabel[cl](20, -20){$
W = \left(
\begin{array}{r r r r}
  -2 & -1 & 1 & 2 \\
  2 & 1 & -1 & -2
\end{array}
\right)\,.
$\\
$m_0 = (2, 0)\,.$
}

\end{mfpic}
\caption{Petri net.}
\label{fig:egPN}
\end{figure}


A marking $m'$ is {\em reachable} from the marking
$m$ if there exists a {\em firing sequence} $\sigma = t_1\dots t_k$ ($k \geq
0$) and a sequence of markings 
$m_1,\dots, m_{k-1}$ such that $m \xrightarrow{t_1} m_1
\xrightarrow{t_2} \cdots \xrightarrow{t_{k-1}} m_{k-1} \xrightarrow{t_k}
m'$. We write in a condensed way: $m \xrightarrow{\sigma} m'$.

We denote by $\mathcal{R}(m)$ the set of markings which are reachable from
the marking
$m$. The {\em reachability graph} of a Petri net with initial marking
$m_0$ is the 
directed graph with nodes $\mathcal{R}(m_0)$ and arcs $\{(m, m') |
\exists t \in T: m 
\xrightarrow{t} m'\}$.

Given  $(\cN,m_0)$ and $m_1$, the {\em reachability problem} is to decide if $m_1\in
\mathcal{R}(m_0)$, and the {\em coverability
  problem} is to decide if $\exists m_2\in
\mathcal{R}(m_0), m_2\geq m_1$.

A Petri net $(\cN,m_0)$  is {\em live} if every transition can always be
enabled again, that is: $\forall m \in \mathcal R(m_0),\forall t \in T, \ \exists m' \in \mathcal
R(m) , \ m' \xrightarrow{t}$.
A Petri net $(\cN,m_0)$  is bounded if $\cR(m_0)$ is finite. It is {\em
  safe} or {\em 1-bounded} if: $\forall m \in \mathcal R(m_0),~\forall p \in P, 
  \ m(p) \leq 1$.  


\subsection{Product-form Petri nets}
\label{subsec:spn}

There exist several ways to define timed models of Petri
nets, see \cite{BCOQ}. We consider the model of
Markovian Petri nets with {\em race policy}. 
Roughly, with each enabled transition is associated a ``countdown
clock'' whose positive initial value is set at random according to an
exponential distribution whose rate depends on the transition. The
first transition to reach 0 fires, which may enable new transitions
and start new clocks. We adopt here the \emph{single-server policy}
which means that the rate of a transition does not depend on 
the enabling degree of the transition. In the more general definition
of product-form Petri nets~\cite[Definition~8]{Haddad05}, rates may depend on the current marking
in a restricted way. For the sake of readability, we have chosen a simpler
version. Results of sections~\ref{section:synthesis} and~\ref{section:qualitative}
still hold with the general definition. On the other hand, it is well-known that the
complexity of the computation of the normalisation constant highly increases
even for the simple case of queuing networks. Here also the results of section~\ref{section:normalization}
are only valid with constant rates.   

\begin{defi}[Markovian PN]\label{def:SPN}
A {\em Markovian Petri net (with race policy)} is a Petri
net equipped with a set
of {\em rates} $(\mu_t)_{t\in T}$, $\mu_t \in \mathbb R_+^*$.
The firing time of an enabled transition $t$ is exponentially distributed
with parameter $\mu_t$.
The marking
evolves as a continuous-time jump Markov process with state space
$\mathcal R(m_0)$ and infinitesimal generator $Q=(q_{m,m'})_{m,m' \in \mathcal
R(m_0)}$, given by:
\begin{equation}\label{eq:Q}
\forall m, \ \forall m' \neq m, \ q_{m, m'} = \sum_{t \mbox { \scriptsize{such that} } m
  \xrightarrow{t} m'} 
\mu_t, \qquad \forall m, \ q_{m,m} = -\sum_{m'\neq m} q_{m,m'}\,. 
\end{equation}
\end{defi}


W.l.o.g., we
assume that there is no transition $t$ such that ${^\bullet t}=
t^\bullet$. Indeed, the firing of such a transition does not modifiy
the marking, so its removal does not modify the infinitesimal
generator. We also assume
that $({^\bullet t_1}, t_1^\bullet) \neq ({^\bullet t_2},
t_2^\bullet)$ for all transitions $t_1\neq t_2$. Indeed, if it is not the
case, the two transitions may be replaced by a single one with the
summed rate.

An {\em invariant measure} is a non-trivial solution $\nu$ to the {\em balance
equations}: $\nu Q = (0,\dots , 0)$. A {\em stationary distribution} $\pi$ is an invariant
probability measure: $\pi Q = (0, \dots , 0)$, $\sum_m \pi(m) = 1$.

\begin{defi}[Product-form PN]\label{de-pfpn}
A Petri net is a {\em product-form Petri net} if for all rates
$(\mu_t)_{t\in T}$, the corresponding Markovian
Petri net admits an
invariant measure $\nu$ satisfying:
\begin{equation}\label{eq-pf}
\exists (u_p)_{p\in P}, u_p \in \mathbb R_+, \quad \forall m \in \cR(m_0),
\qquad \nu(m) = 
\prod_{p\in P} u_p^{m_p} \:. 
\end{equation}
\end{defi}

The existence of $\nu$ satisfying \eref{eq-pf} implies that the marking
process is irreducible (in other words, the reachability graph is
strongly connected). In \eref{eq-pf}, the mass of the measure, i.e. $\nu(\cR(m_0))=\sum_m
\nu(m)$, may be either finite or infinite. 
For a bounded Petri net, the mass is always finite. But for an
unbounded Petri net, the typical situation will be as follows: structural
conditions on the Petri net will ensure that the Petri net is a
product-form one. Then, for some values of the rates, $\nu$ will have an infinite mass,
and, for others, $\nu$ will have a finite mass. In the first situation, the
marking process will be either transient or recurrent null (unstable
case). In the second situation, the marking process will be positive
recurrent (stable or ergodic case). 

When the mass is finite, we call $\nu(\cR(m_0))$ the {\em normalising
  constant}. The probability measure $\pi(\cdot) =
\nu(\cR(m_0))^{-1}\nu(\cdot)$ is the unique stationary measure of the marking
process. Computing explicitly the normalising
constant is an important issue, see Section \ref{section:normalization}. 

\medskip

The goal is now to get sufficient conditions for a Petri net to be of
product-form. To that purpose, we
introduce three notions: {\em weak reversibility}, {\em
  deficiency}, and {\em witnesses}. 

\medskip

Let $(N, m_0)$ be a Petri net.
  The set of {\em complexes} is defined by $\cC=\{{}^\bullet t \mid t \in T\} 
  \cup \{t^\bullet \mid t \in T\}$.
  The {\em reaction graph} is the directed graph whose set of nodes is
  $\cC$ and whose set of arcs is $\{({^\bullet}t, t^\bullet) | t \in
  T\}$.  It can be viewed as a state machine.

\begin{defi}[Weak reversibility: $\piun$-nets]
  A Petri net is {\em weakly reversible (WR)} if every connected
  component of its 
  reaction graph is strongly connected.
  Weakly reversible Petri nets are also
 called {\em $\piun$-nets}.\label{def-wr}
\end{defi}

The notion and the name ``WR'' come from the chemical literature. In the
Petri net context, it was introduced in \cite[Assumption
3.2]{bouchserercim} under a different name and with a slightly different but equivalent
formulation. WR is a strong
constraint. It should not be confused with the classical notion of 
``reversibility'' (the marking graph is strongly connected). In
particular WR, a structural property,  implies reversibility, a behavioural one! 
Observe that all symmetric Petri nets are WR.  

\smallskip

The notion of deficiency is due to Feinberg~\cite{Feinberg79}.





\begin{defi}[Deficiency]
Consider a Petri net with incidence matrix $W$ and set of complexes
$\cC$. Let $\ell$ be
the number of connected components of the reaction graph. 
  The {\em deficiency} of the Petri net is defined by: $
  |\cC| -  \ell - \mbox{rank}(W)$. 
\end{defi}

The notion of witnesses appears in \cite{Haddad05}.

\begin{defi}[Witness]
  Let $c$ be a complex. A {\em witness} of $c$ is a vector $wit(c)\in
  \mathbb Q^P$ such that for all transition $t$:
\[
\begin{cases}
wit(c)\cdot W(t) =-1 & \mbox{if } {}^\bullet t=c \\
wit(c)\cdot W(t) =1 & \mbox{if } t^\bullet=c \\
wit(c)\cdot W(t) =0 & \mbox{otherwise}\:,
\end{cases}
\]
where $W(t)$ denotes the column vector of $W$ indexed by $t$.
\end{defi}

 \noindent{\bf Examples.} Consider the Petri net of Figure
 \ref{fig:egPN}. First, it is WR. Indeed, 
 the set of complexes is $\cC = \{p_1, p_2, 2p_1, 2p_2\}$ and the reaction
 graph is: 
 \[
 p_1 \leftrightarrow p_2\,,\ 2p_1 \leftrightarrow
 2p_2\,,
 \]
 with two connected components which are strongly
 connected. Second, the deficiency is 1 since $|\cC|=4$, $\ell =2$,
 and $\mbox{rank}(W)=1$. Last, one can check that none of the
 complexes admit a witness. 

 The Petri net of Figure~\ref{fig:nodota} is WR and has deficiency 0. 
 Note that the witnesses may not be unique.
 Possible witnesses are: $wit(2p_1 + q_1) = q_1$, $wit(p_1 + q_2) = q_2$, $wit(p_2 + q_3)
 = q_3$, $wit(2p_2 + q_4) = q_4$. 
 Another possible set of witnesses is $\{q_1, q_2, -q_2, -q_1\}$.


\begin{prop}[deficiency 0 $\iff$ witnesses, in \protect{\cite[Prop.
  3.9]{MairesseNguyen09}}]\label{pr-equiv}
A Petri net admits a witness for each complex iff it has deficiency 0. 
\end{prop}

Next Theorem is a combination of Feinberg's Deficiency
zero Theorem~\cite{Feinberg79} and Kelly's
Theorem~\cite[Theorem 8.1]{kell79}. (It is proved under this form in
\cite[Theorem 3.8]{MairesseNguyen09}.) 

\begin{theo}[WR + deficiency 0 $\implies$ product-form]
\label{thm:defzero} 
Consider a Markovian Petri net with rates $(\mu_t)_{t\in
    T}$, $\mu_t >0$, and assume that the underlying Petri net is WR and has
deficiency 0. 
Then there exists $(u_p)_{p\in P}$, $u_p>0$, satisfying the equations:
  \begin{equation}
\forall c \in \cC, \qquad     \prod_{p: c_p \neq 0} u_p^{c_p} \sum_{t: ^\bullet{t} = c}\mu_t  = \sum_{t:
    t^\bullet = c}\mu_{t} \prod_{p: ^\bullet{t}_p \neq
      0}u_p^{^\bullet{t}_p}\,. \label{eq:NLTE} 
  \end{equation}
  The marking process has an invariant measure $\nu$ such that:
  \[\forall m,
  \ \nu(m)= \Phi(m)^-1~\prod_{p \in P} u_p^{m_p}\,.\]
\end{theo}



Checking the WR, computing the deficiency, determining the witnesses,
and solving the equations \eref{eq:NLTE}, all of these operations can
be performed in polynomial-time, see \cite{Haddad05,MairesseNguyen09}.

Summing up the above, it seems worth to isolate and christen the
class of nets which are WR and have deficiency 0. We adopt the
terminology of \cite{Haddad05}. 

\begin{defi}[$\pideux$-net]
A {\em $\pideux$-net} is a Petri net
which is WR and has deficiency 0. 
\end{defi}

\subsection{Some properties of WR and deficiency zero nets}\label{subsec:liveboundedpi2}

Let $\cN = (P, T, W^-, W^+)$ be a
Petri net. Let $W=W^+-W^-$ be 
the incidence matrix of $\cN$ and let $A$ be the incidence matrix of
the reaction graph.

\medskip

Consider at first free-choice nets. It was shown in \cite[Section 4.3]{MairesseNguyen09} that 
for free-choice nets, WR implies deficiency zero. The converse does not hold for
general free-choice nets. For instance, state machines always have deficiency zero~\cite[Prop.
3.2]{MairesseNguyen09}, and may not be WR. For marked graphs, however,
the converse is true, and stated below. 

\begin{prop}\label{WRmarkedgraph}
  The deficiency of a connected marked graph is either 0 or 1. A marked graph has
  deficiency zero if and only if it is WR.
\end{prop}

\begin{proof}
  Let $\cN$ be a marked graph. According to \cite[Prop. 3.16]{DeEs95}, the only
  T-semi-flows of $\cN$ are $a(1,\cdots,1)$, $a \in \Q$, hence
  $\mbox{rank}(W) = |T| - 1$. Since $A$ is a $\cC \times T$ matrix,
  $\mbox{rank}(A) \le |T|$. Hence $\delta = \mbox{rank}(A) - \mbox{rank}(W) \le
  1$.

  \medskip

  \noindent
  The ``if'' direction of the second claim is trivial since a marked graph is a
  free-choice net. Consider the ``only if'' direction. Let $\cN$ be a deficiency
  zero marked graph. Let ${\bf 1}$ be the column vector $(1,\dots ,1)$
  of size $T$. Since $\cN$ is a marked graph, we have $W\cdot {\bf 1}
  =(0,\dots , 0)$. 
  By Proposition \ref{pr-equiv}, $A = BW$ for
  some $\Q$-valued matrix $B$. So we have $A\cdot {\bf 1} = BW\cdot
  {\bf 1}= (0,\dots, 0)$. This implies that the connected components of the
  reaction graph must be strongly connected. 
  Indeed pick a connected component which is not strongly connected. It admits
  a partition of its complexes into two subsets $C_1$ and $C_2$
  such that there at least one transition $t$ from $C_1$ to $C_2$ and no transition
  from $C_2$ to $C_1$. Then vector $x$ defined by $x(c)=0$ for $c \in C_1$
  and $x(c)=1$ for $c \in C_2$ fulfills $x.A\geq 0$ and $x.A(t)> 0$. Thus $x.A.{\bf 1}>0$
  yields a contradiction. So $\cN$ is WR.
\end{proof}


\begin{prop}\label{pr-lb}
  For a live and bounded Petri net, deficiency zero implies weak
  reversibility.
\end{prop}

\begin{proof}
  Let $m_o$ be a marking such that $(\cN,m_0)$ is live and bounded. We
  assume that $\cN$ has deficiency 0 but is not WR. 
  Then there exists a terminal
  strongly connected component $C$ of the reaction graph and a transition $t_0$
  such that $t_0^{\bullet} \in C$ and ${}^{\bullet}t_0 \notin C$.\\
  We claim that for every vector $v \in \Q^{T}$ such that for all $t \in T$,
  $v(t)\geq 0$ and $v(t_0)>0$, we have $Av \neq (0,\dots, 0)$. Indeed,
  \begin{eqnarray*}
	\sum_{c \in C}(Av)(c) &=& \sum_{c \in C} \left( \sum_{t \in T} v(t)
	\bigl({\bf 1}_{t^{\bullet} = c} - {\bf 1}_{{}^{\bullet}t = c}  \bigr)  \right)\\
	&=& \sum_{t \in T -\{t_0\}} v(t) \left( \sum_{c \in C} \bigl(  {\bf
	1}_{t^{\bullet} = c} - {\bf 1}_{{}^{\bullet}t = c} \bigr)  \right) +
	v(t_0)\:.
  \end{eqnarray*}
  Since $C$ is a terminal strongly connected component, $\sum_{c\in C}
  {\bf 1}_{t^{\bullet} = c} -
  {\bf 1}_{{}^{\bullet}t = c}$ is either $0$ or $1$ for all $t \in T$. 
  Hence $\sum_{c \in C}(Av)(c) \ge v(t_0) > 0$. The claim is proved.

  \medskip

  \noindent
  Since $(\cN, m_0)$ is live and bounded, there exists a strictly positive
  T-semi-flow $v \in \Q^T$~\cite[Theorem 2.38]{DeEs95}, that is: 
  $\forall t,  \ v(t)>0, \ W\cdot v= (0,\dots , 0)$. Now recall that the
  deficiency of $\cN$ is 0. According to  Proposition~\ref{pr-equiv},
  there exists a $\cC \times P$ matrix $B$ such
  that $A = BW$. We get $Av = BWv = (0,\dots , 0)$. This
  contradicts the above claim.
\end{proof}

A {\em home marking} is a marking which is reachable from every reachable
marking. Having a home marking is an important property for Markovian
Petri nets.  Indeed, a Petri net has a home
marking iff its reachability graph has only one terminal strongly
connected component. And this last condition is required for the
marking process to be ergodic.

\begin{prop}\label{def0hmwr}
Let $\cN$ be a deficiency zero Petri net. Then $\cN$ is WR iff 
there exists a marking $m_0$ such that $(\cN,m_0)$ is live and 
$m_0$ is a home marking.
\end{prop}

\begin{proof}
  Suppose that $\cN$ is WR. Let $m_0$ be a marking which enables every
  transition.
  The definition of weak reversibility implies that every arc
  of the reachability graph belongs to a cycle, so the reachability graph is
  strongly connected, that is $m_0$ is a home marking. The liveness follows
  trivially.

  \medskip

  \noindent
  Now suppose that there exists a marking $m_0$ such that $(\cN,m_0)$ is live and 
  $m_0$ is a home marking but $\cN$ is
  not WR. We proceed as in the proof of Prop. \ref{pr-lb}.  Let $C$ be
  a terminal
  strongly connected component of the reaction graph and let $t$ be a
  transition such that $t^{\bullet} \in C$ and ${}^{\bullet}t \notin C$. 
  Since $(\cN, m_0)$ is live there is a path $\gamma_1$ in the reachability graph
  from $m_0$ to $m_1$ which enables $t$. Let $m'_1$ be the marking reached by the firing
  of $t$, since $m_0$ is a home marking there is a path $\gamma_2$ from $m'_1$
  to $m_0$. Thus $\gamma =\gamma_1 t \gamma_2$ is a (directed) cycle  
  of the reachability graph of $(\cN, m_0)$. Let $v$ be the $\setN^T$
  column vector such
  that: $\forall u \in T$, $v(u)$ is the number of occurrences of $u$ in
  $\gamma$. Clearly, $v(t) > 0$ and $W.v = (0,\dots , 0)$. The end of
  the argument follows from the claim inside the proof of Prop.~\ref{pr-lb}.
\end{proof}


The interest of Prop. \ref{def0hmwr}
is twofold. On the one hand, it connects weak reversibility and deficiency
zero which are two independent properties (\cite{MairesseNguyen09}). On the
other hand, it shows that the only deficiency zero and live Markovian Petri
nets which are ergodic are the \pideux-nets. 

\medskip

Figure \ref{fig:defandwr} recapitulates the relations between deficiency and
weak reversibility. The shaded cells correspond to impossibilities. For
instance, no WR free-choice nets have strictly positive deficiency.

\begin{figure}[htbp]
\centering
\begin{mfpic}{-140}{140}{-40}{40}
\tlabelsep{2pt}

\rect{(0, 0), (40, -20)}
\shiftpath{(40, 0)}\rect{(0, 0), (40, -20)}
\shiftpath{(80, 0)}\rect{(0, 0), (40, -20)}
\shiftpath{(0, -20)}\rect{(0, 0), (40, -20)}
\shiftpath{(40, -20)}\rect{(0, 0), (40, -20)}
\shiftpath{(80, -20)}\rect{(0, 0), (40, -20)}
\shiftpath{(0, -40)}\rect{(0, 0), (40, -20)}
\draw[black]\rhatch[2pt][black]\shiftpath{(40, -40)}\rect{(0, 0), (40, -20)}
\draw[black]\rhatch[2pt][black]\shiftpath{(80, -40)}\rect{(0, 0), (40, -20)}

\tlabel[cc](60, -10){\scriptsize{WR}}
\tlabel[cc](100, -10){\scriptsize{Not WR}}
\tlabel[cc](20, -30){\scriptsize{$\delta = 0$}}
\tlabel[cc](20, -50){\scriptsize{$\delta > 0$}}

\tlabel[cc](60, -70){State machines}

\shiftpath{(160, 0)}\rect{(0, 0), (40, -20)}
\shiftpath{(160, 0)}\shiftpath{(40, 0)}\rect{(0, 0), (40, -20)}
\shiftpath{(160, 0)}\shiftpath{(80, 0)}\rect{(0, 0), (40, -20)}
\shiftpath{(160, 0)}\shiftpath{(0, -20)}\rect{(0, 0), (40, -20)}
\shiftpath{(160, 0)}\shiftpath{(40, -20)}\rect{(0, 0), (40, -20)}
\draw[black]\rhatch[2pt][black]\shiftpath{(160, 0)}\shiftpath{(80, -20)}\rect{(0, 0), (40, -20)}
\shiftpath{(160, 0)}\shiftpath{(0, -40)}\rect{(0, 0), (40, -20)}
\draw[black]\rhatch[2pt][black]\shiftpath{(160, 0)}\shiftpath{(40, -40)}\rect{(0, 0), (40, -20)}
\shiftpath{(160, 0)}\shiftpath{(80, -40)}\rect{(0, 0), (40, -20)}

\tlabel[cc](220, -10){\scriptsize{WR}}
\tlabel[cc](260, -10){\scriptsize{Not WR}}
\tlabel[cc](180, -30){\scriptsize{$\delta = 0$}}
\tlabel[cc](180, -50){\scriptsize{$\delta > 0$}}

\tlabel[cc](220, -70){Marked graphs}

\shiftpath{(0, -90)}\rect{(0, 0), (40, -20)}
\shiftpath{(0, -90)}\shiftpath{(40, 0)}\rect{(0, 0), (40, -20)}
\shiftpath{(0, -90)}\shiftpath{(80, 0)}\rect{(0, 0), (40, -20)}
\shiftpath{(0, -90)}\shiftpath{(0, -20)}\rect{(0, 0), (40, -20)}
\shiftpath{(0, -90)}\shiftpath{(40, -20)}\rect{(0, 0), (40, -20)}
\shiftpath{(0, -90)}\shiftpath{(80, -20)}\rect{(0, 0), (40, -20)}
\shiftpath{(0, -90)}\shiftpath{(0, -40)}\rect{(0, 0), (40, -20)}
\draw[black]\rhatch[2pt][black]\shiftpath{(0, -90)}\shiftpath{(40, -40)}\rect{(0, 0), (40, -20)}
\shiftpath{(0, -90)}\shiftpath{(80, -40)}\rect{(0, 0), (40, -20)}

\tlabel[cc](60, -100){\scriptsize{WR}}
\tlabel[cc](100, -100){\scriptsize{Not WR}}
\tlabel[cc](20, -120){\scriptsize{$\delta = 0$}}
\tlabel[cc](20, -140){\scriptsize{$\delta > 0$}}

\tlabel[cc](60, -160){Free-choice nets}

\shiftpath{(160, -90)}\rect{(0, 0), (40, -20)}
\shiftpath{(160, -90)}\shiftpath{(40, 0)}\rect{(0, 0), (40, -20)}
\shiftpath{(160, -90)}\shiftpath{(80, 0)}\rect{(0, 0), (40, -20)}
\shiftpath{(160, -90)}\shiftpath{(0, -20)}\rect{(0, 0), (40, -20)}
\shiftpath{(160, -90)}\shiftpath{(40, -20)}\rect{(0, 0), (40, -20)}
\draw[black]\rhatch[2pt][black]\shiftpath{(160, -90)}\shiftpath{(80, -20)}\rect{(0, 0), (40, -20)}
\shiftpath{(160, -90)}\shiftpath{(0, -40)}\rect{(0, 0), (40, -20)}
\draw[black]\rhatch[2pt][black]\shiftpath{(160, -90)}\shiftpath{(40, -40)}\rect{(0, 0), (40, -20)}
\draw[black]\rhatch[2pt][black]\shiftpath{(160, -90)}\shiftpath{(80, -40)}\rect{(0, 0), (40, -20)}

\tlabel[cc](220, -100){\scriptsize{WR}}
\tlabel[cc](260, -100){\scriptsize{Not WR}}
\tlabel[cc](180, -120){\scriptsize{$\delta = 0$}}
\tlabel[cc](180, -140){\scriptsize{$\delta > 0$}}

\tlabel[cc](220, -160){Live and bounded nets}
\tlabel[cc](220, -170){and nets which have a live home marking}

\end{mfpic}
\caption{Relations between  deficiency ($\delta$) and WR for some
  classes of Petri nets.}
\label{fig:defandwr}
\end{figure}

\section{Synthesis and regulation of \pideux-nets}
\label{section:synthesis}

The reaction graph, defined in
Section \ref{subsec:spn}, may
be viewed as a Petri net (state machine). Let us formalise
this observation.  
\noindent
The \emph{reaction Petri net} of $\cN$
is the Petri net $\cA = (\mathcal C, T,
\overline{W}^-, \overline{W}^+)$,  
with for every $t \in T$: 
\begin{itemize}[nolistsep]
	\item $\overline{W}^-({}^\bullet t,t)= 1 \mbox{ and }   \forall u \neq {}^\bullet
t,\ \overline{W}^-(u,t)= 0$ 
	\item $\overline{W}^+(t^\bullet,t)= 1\mbox{ and }  \forall u \neq
t^\bullet,\ \overline{W}^+(u,t)= 0$  
\end{itemize}


\subsection{Synthesis}
\label{subsec:synthesis}

In this subsection, we consider unmarked nets. 
We define
three rules that generate all the $\pideux$-nets. The first
rule  adds a strongly connected state machine.

\begin{defi}[State-machine insertion]
  Let $\mathcal N=(P_\mathcal N,T_\mathcal N,W^-_\mathcal N,W^+_\mathcal N)$ be a net
  and $\mathcal M=(P_\mathcal M,T_\mathcal M,W^-_\mathcal M,W^+_\mathcal M)$ be a
  strongly connected state machine disjoint from $\mathcal N$.
  The rule {\tt S-add} 
  is always applicable and  $\mathcal N'=\mbox{{\tt S-add}}(\mathcal N,\mathcal M)$
  is defined by:
  \begin{itemize}[nolistsep]
    \item $P'= P_\mathcal N\sqcup P_\mathcal M$, $T'=T_\mathcal N\sqcup T_\mathcal M$;
    \item $\forall  p \in P_\mathcal N, \  \forall  t \in T_\mathcal N, \ W'^-(p,t)=W^-_\mathcal
      N(p,t), \ W'^+(p,t)=W^+_\mathcal N(p,t)$;
    \item $\forall p  \in P_\mathcal M, \   \forall  t \in T_\mathcal M, \ W'^-(p,t)=W^-_\mathcal
      M(p,t), \ W'^+(p,t)=W^+_\mathcal M(p,t)$;
    \item All other entries of $W'^-$ and $W'^+$ are null.
  \end{itemize}
\end{defi}

The second rule consists in substituting to a complex $c$ the complex
$c+\lambda p$. However 
in order to be applicable some conditions must be fulfilled. 
The first condition requires that 
$c(p)+\lambda$ 
is non-negative. The second condition ensures that the substitution does not modify
the reaction graph. The third condition preserves deficiency zero. Observe that the third
condition can be checked in polynomial time, indeed it amounts to
solving a system of linear equations in $\mathbb{Q}$ for every
complex. 

\begin{defi}[Complex update]
  Let $\mathcal N=(P,T,W^-,W^+)$ be a \pideux-net, $c$ be a complex of $\mathcal N$, $p \in P$,
  $\lambda \in \setZ \setminus \{0\}$.
  The rule {\tt C-update} 
  is applicable when:
  \begin{enumerate}[nolistsep]
    \item $\lambda+c(p) \geq 0$;
    \item $c+\lambda p$ is not a complex of $\mathcal N$;
    \item For every complex $c'$ there exists a witness $wit(c')$ s.t. $wit(c')(p)=0$.
  \end{enumerate}
  The resulting net $\mathcal N'=\mbox{{\tt C-update}}(\mathcal N,c,p,\lambda)$ 
  is defined by:
  \begin{itemize}[nolistsep]
    \item $P'= P$,  $T'=T$;
    \item \begin{small} 
	$\forall t \in T\  \mbox{s.t.}\ W^-(t)\neq c, \ W'^-(t)=W^-(t)$,  $\forall t \in T\  \mbox{s.t.}\ W^-(t)= c,\ W'^-(t)=c+\lambda p$ 
      \end{small}
    \item \begin{small} $\forall t \in T\  \mbox{s.t.}\ W^+(t)\neq
	c,\ W'^+(t)=W^-(t)$, $\forall t \in T\  \mbox{s.t.}\ W^+(t)=
	c,\ W'^+(t)=c+\lambda p$. 
      \end{small} 
  \end{itemize}
\end{defi}

The last rule ``cleans'' the net by deleting an isolated place.
We call this operation {\tt P-delete}.

\begin{defi}[Place deletion]
  Let $\mathcal N=(P,T,W^-,W^+)$ be a net and let $p$ be an isolated place
  of $\mathcal N$, i.e.  $W^-(p)=W^+(p)=0$. 
  Then the rule {\tt P-delete}  is 
  applicable and  $\mathcal N'=\mbox{ {\tt P-delete}}(\mathcal N,p)$ is
  defined by: 
  \begin{itemize}[nolistsep]
    \item $P'= P\setminus \{p\}$, $T'=T$;
    \item $\forall q \in P',\ W'^-(q)=W^-(q),\ W'^+(q)=W^+(q)$.
  \end{itemize}
\end{defi}

Proposition~\ref{pr-sc} shows the interest of the rules 
for synthesis of $\pideux$-nets.
\begin{prop}[Soundness and Completeness]\label{pr-sc}
  Let $\mathcal N$ be a $\pideux$-net.
  \begin{itemize}[nolistsep]
    \item If a rule  {\tt S-add},  {\tt C-update} or {\tt P-delete}
      is applicable on $\mathcal N$ then the resulting net is still a $\pideux$-net.
    \item The net $\mathcal N$ can be obtained by successive applications of
      the rules 
      {\tt S-add},  {\tt C-update}, {\tt P-delete} starting from the empty net.  
  \end{itemize}  
\end{prop}

\begin{proof}
  {\bf Soundness.} The case  of   {\tt P-delete} is straightforward. Since we delete an isolated
  place, the reaction graph is unchanged. So the net is still WR. Assume that we
  delete an isolated place $p$
  and that $p$ occurs in a witness $wit(c)$ of some complex $c$. Then
  $wit(c)-wit(c)(p)$ is also a witness of $c$.

  \smallskip \noindent
  Let us examine the application of rule $\mbox{{\tt S-add}}(\mathcal N,\mathcal M)$.
  The state machine $\mathcal M$
  constitutes a new component of the reaction graph. Since $\mathcal M$ is strongly connected, the new net
  is still WR. The witness of complexes associated with $\mathcal N$ are unchanged.
  Let $q$ be a place of $\mathcal M$; by definition of state machines this
  place is self-witnessing i.e. $wit(q)=q$.
  Thus the new net has deficiency zero.

  \smallskip \noindent
  Let us examine the application of  the rule $\mbox{{\tt C-update}}(\mathcal N,c,p,\lambda)$. 
  By the second condition of its application the reaction graph of the new net is the same
  as the original one (with $c+\lambda p$ instead of $c$). So the new net is WR.
  Due to the third condition, the witness of $c'\neq c$ is unchanged and the witness of
  $c+\lambda\cdot p$ is the one of $c$.

  \medskip

  \noindent{\bf Completeness.}
  Let $\mathcal N=(P,T,W^-,W^+)$ be a $\pideux$-net. We proceed as follows to
  generate $\cN$ via our rules.
  At any stage of the generation,  $\mathcal N_{cur}$ denotes the current net. Initially
  $\mathcal N_{cur}$ is the empty net.

  \smallskip \noindent {\bf First step.} 
  Let $\mathcal A_1,\ldots,\mathcal A_n$ be the strongly connected state machines corresponding to the components
  of the reaction net of $\mathcal N$. 
  Given a complex $c$ of $\mathcal N$, the corresponding place in the
  state machine is denoted $q_c$. We apply the rules $\mbox{{\tt S-add}}(\mathcal N_{cur},\mathcal A_i)$
  for $i$ from 1 to $n$. 
  At this stage, $\mathcal N_{cur}$ has $T$ for set of transitions
  and a place $q_c$ for every complex $c$ of $\mathcal N$. Furthermore, $q_c$ has for input (resp. output) transitions
  the input  (resp. output) transitions of $c$ in $\mathcal N$. The
  complexes of $\mathcal N_{cur}$ are the places $q_c$ and they
  are their own witnesses.

  \smallskip \noindent {\bf Second step.} It consists in adding the places of $P$ in such a way that
  the net $\mathcal N_{cur}$ restricted to the places of $P$ is $\mathcal N$.
  At every stage of this step,
  given a complex $c=\sum_{p \in P} c(p)p$ of $\mathcal N$, there is a corresponding complex 
  $c' = q_c + \sum_{p \in P\cap P_{cur}} c(p)p$ in $\mathcal N_{cur}$. 
  For every place $p \in P$, we add $p$ to $\mathcal N_{cur}$ by rule {\tt S-add} (an isolated place is
  a strongly connected state machine) and for every
  complex $c$ of $\mathcal N$ such that $c(p)>0$, we apply the rule $\mbox{{\tt C-update}}(\mathcal N_{cur},c',p,c(p))$.
  Let us check that this rule is applicable. First, $c'(p)+c(p)=c(p)$ is
  positive. Second, $c'+c(p)p$ is not a complex of $\mathcal N_{cur}$ by
  construction. 
  Third, for every complex $c'$
  of $\mathcal N_{cur}$, there is a witness consisting in the single
  place $q_c$ which is in a state machine $\mathcal A_i$
  (thus different from $p$).
  At the end of this step, $\mathcal N_{cur}$ is the net $\mathcal N$
  enlarged with  the places of the state machines $\cA_i$.
  Otherwise stated, every complex $c'$ of $\mathcal N_{cur}$ is equal to $c+q_c$.

  \smallskip \noindent {\bf Third step.} This step consists in deleting
  the places of the state machines.
  We observe that the place $q_c$ only occurs in the complex $c+q_c$. 
  The net $\mathcal N$ being a $\pideux$-net, every complex $c'$ has a witness
  $wit(c')$ in $\cN$. Then $wit(c')$ is a witness 
  for $c' +q_{c'}$ in $\mathcal N_{cur}$ whose support does not contain $q_c$. 
  Thus the rule $\mbox{{\tt C-update}}(\mathcal N_{cur}, c+q_c, q_c, -1)$
  is applicable. 
  After its application, $q_c$ becomes isolated and can be deleted by
  the rule $\mbox{{\tt P-delete}}(\mathcal N_{cur},q_c)$. At the end,
  we have obtained $\mathcal N$.
\end{proof}

\begin{figure}[hpbt]
  \begin{center}
    \includegraphics[scale=0.4]{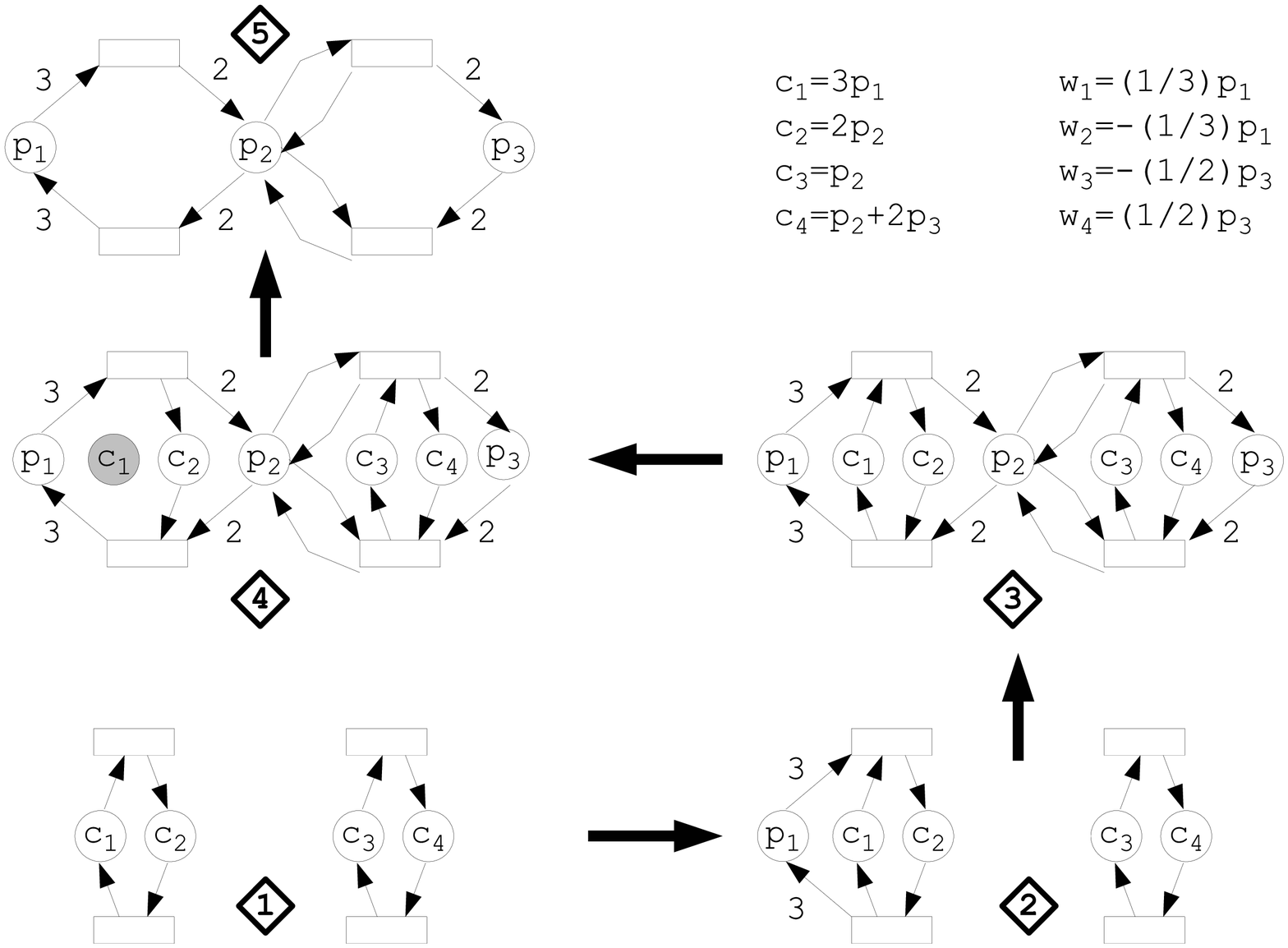}
  \end{center}  
  \caption{How to synthetise a $\pideux$-net.}
  \label{fig:synthnet}
\end{figure}

\noindent{\bf Example.}
We illustrate the synthesis process using our rules on the net numbered 5
in Figure~\ref{fig:synthnet}. We have also indicated on the right
upper part of this figure, the four complexes and their witnesses. Since
the reaction Petri graph of this net has two state machines, we start by creating it using twice the insertion
of a state machine (net 1).
Then we add the place $p_1$ (a particular state machine). We update the complex
$c_1$ (the single one where $p_1$ appears in the original net) by adding $3p_1$ (net 2).
Iterating
this process, we obtain the net 3.
Observe that this net is a fusion (via $T$ the set of transitions)
of the original net and its reaction Petri net. We now iteratively update the
complexes. The net 4 is the result of transforming $c_1+3p_1$
into $3p_1$. 
Once $c_1$ is isolated, we delete it. Iterating this process yields the original
net.

\smallskip
For modelling purposes, we could define
more general rules like the refinement of a place by a strongly connected state machine. 
Here the goal was to design a minimal set of rules.

\subsection{From non $\pideux$-nets to $\pideux$-nets}

Below we propose a procedure which takes as input any
Petri net and returns a $\pideux$-net. The important
disclaimer is that the resulting net, although related to the original
one, has a different structural and timed behaviour. So it is up to
the modeller to decide if the resulting net satisfies the desired
specifications. In case of a positive answer, the clear gain is that
all the associated Markovian Petri nets have a product form. 

\smallskip

Consider a Petri net $\cN =(P, T, W^-,W^+,m_0)$ with set of complexes
$\cC$. Assume that $\cN$ is not WR. For each transition $t$, add a
reverse transition $t^-$ such that ${}^\bullet t^- = t^{\bullet}$ and
$(t^-)^\bullet = {}^{\bullet}t$ (unless such a transition already
exists). The resulting net is WR. In the Markovian Petri net, the
added reverse transitions can be given very small rates, to
approximate more closely the original net. However, there is no 
theoretical guarantee of the convergence of steady-state distributions
and in fact counter-examples can be exhibited. 

\smallskip

Now, to enforce deficiency 0, the
idea is to compose a general Petri net with its reaction graph as in
the illustration of Proposition \ref{pr-sc}.


\begin{defi}\label{def:control}
Consider a Petri net $\cN=(P,T, W^-,W^+,m_0)$. Let $\overline{m}_0$ be an initial marking
for the reaction Petri net $\cA$.
The {\em regulated} Petri net associated with $\cN$ is defined as
follows: 
\[
\cA \odot \cN = \bigl( P
\sqcup \mathcal C, 
T, \widetilde{W}^-, \widetilde{W}^+ , (m_0,\overline{m}_0) \bigr),
\quad  \widetilde{W}^- = \left[ \begin{array}{c} W^-
    \\ \overline{W}^- \end{array} \right]\:,  \widetilde{W}^+ = \left[ \begin{array}{c} W^+
    \\ \overline{W}^+ \end{array} \right]\:.
\]
\end{defi}



\begin{prop}\label{prop-regul}
The regulated Petri net $\mathcal A \odot \mathcal
N$ is WR iff $\mathcal N$ is WR. 
The regulated Petri net $\mathcal A \odot \mathcal N$ has deficiency 0.  
\end{prop}

\begin{proof}
  By construction the reaction graph of the regulated Petri net $\mathcal A \odot \mathcal
  N$  is the reaction graph of $\mathcal N$, i.e. $\mathcal A$, modulo a node
  renaming. So $\mathcal A \odot \mathcal
  N$ is WR iff $\mathcal N$ is WR.\\
  Now let us prove that the deficiency is 0. We use the characterization
  by witnesses, see Prop. \ref{pr-equiv}. 
  Let $\widetilde{\cC}$ be the set of complexes of $\mathcal A \odot \mathcal
  N$. 
  Consider $\tilde{c}\in \widetilde{\cC}$ and let $c$ be the
  corresponding element in $\cC$. Define $wit(\tilde{c}) \in \mathbb
  Q^{P\sqcup \cC}$ by:  $wit(\tilde{c})_c =1$, $\forall u \neq c,
  wit(\tilde{c})_u =0$. By direct inspection, we check that
  $wit(\tilde{c})$ is indeed a witness of $\tilde{c}$.
\end{proof}

\begin{figure}[H]
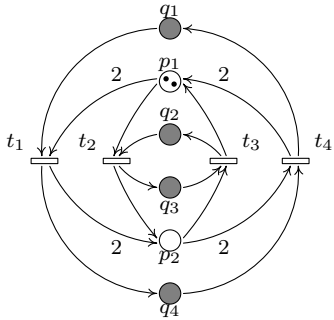

\centering
\begin{mfpic}{-140}{140}{-40}{40}
\tlabelsep{2pt}

\circle{(0, 30), 4}
\tlabel[bc](0, 32){{\scriptsize $p_1$}}
\circle{(0, -30), 4}
\tlabel[tc](0, -32){{\scriptsize $p_2$}}
\draw[black]\gfill[gray(0.5)]\circle{(0, 10), 4}
\tlabel[bc](0, 13){{\scriptsize $q_2$}}
\draw[black]\gfill[gray(0.5)]\circle{(0, -10), 4}
\tlabel[tc](0, -13){{\scriptsize $q_3$}}
\draw[black]\gfill[gray(0.5)]\circle{(0, 50), 4}
\tlabel[bc](0, 52){{\scriptsize $q_1$}}
\draw[black]\gfill[gray(0.5)]\circle{(0, -50), 4}
\tlabel[tc](0, -52){{\scriptsize $q_4$}}

\shiftpath{(-47, 0)}\rect{(-5, -1), (5, 1)}
\shiftpath{(-20, 0)}\rect{(-5, -1), (5, 1)}
\shiftpath{(20, 0)}\rect{(-5, -1), (5, 1)}
\shiftpath{(47, 0)}\rect{(-5, -1), (5, 1)}
\tlabel[br](-52, 2){{\scriptsize $t_1$}}
\tlabel[br](-25, 2){{\scriptsize $t_2$}}
\tlabel[bl](25, 2){{\scriptsize $t_3$}}
\tlabel[bl](52, 2){{\scriptsize $t_4$}}


\arrow\arc[s]{(-5, 31), (-45, 2), 60}
\arrow\arc[s]{(-45, -2), (-5, -31), 60}
\arrow\arc[s]{(5, -31), (45, -2), 60}
\arrow\arc[s]{(45, 2), (5, 31), 60}

\arrow\arc[s]{(-5, 29), (-21, 2), 20}
\arrow\arc[s]{(-21, -2), (-5, -29), 20}
\arrow\arc[s]{(5, -29), (21, -2), 20}
\arrow\arc[s]{(21, 2), (5, 29), 20}

\arrow\arc[s]{(-5, 50), (-48, 2), 90}
\arrow\arc[s]{(-48, -2), (-5, -50), 90}
\arrow\arc[s]{(5, -50), (48, -2), 90}
\arrow\arc[s]{(48, 2), (5, 50), 90}

\arrow\arc[s]{(-5, 10), (-19, 2), 40}
\arrow\arc[s]{(-19, -2), (-5, -10), 40}
\arrow\arc[s]{(5, -10), (19, -2), 40}
\arrow\arc[s]{(19, 2), (5, 10), 40}

\tlabel[br](-16, 28){{\scriptsize $2$}}
\tlabel[tr](-16, -28){{\scriptsize $2$}}
\tlabel[bl](16, 28){{\scriptsize $2$}}
\tlabel[tl](16, -28){{\scriptsize $2$}}

\point[2pt]{(-1.5, 31), (1.5, 29)}

\end{mfpic}
\caption{Regulated Petri net associated with the Petri net of Fig \ref{fig:egPN}.}
\label{fig:nodota}
\end{figure}

The behaviours of the original and regulated Petri nets are
different. In particular, the regulated Petri net is bounded, even if the
original Petri net is unbounded. 
Roughly, the regulation imposes some control on the firing sequences. Consider the example of Figures
\ref{fig:egPN} (original net) and \ref{fig:nodota} (regulated net). The places
$q_1, q_2, q_3, q_4$ correspond to the complexes $2p_1, p_1, p_2, 2p_2$,
respectively.
The transitions $t_1$ and $t_4$ belong to the same simple circuit in
the reaction graph. Let $w$ be an arbitrary firing sequence. The
quantity $|w|_{t1} - |w|_{t4}$ is unbounded for the original net, and
bounded for the regulated net.

\section{Complexity analysis of \pideux-nets}
\label{section:qualitative}

All the nets that we build in this section
are symmetric hence WR. 
For every depicted
transition $t$, the reverse transition exists (sometimes implicitly)
and is denoted $t^-$. 
It is well known that reachability and liveness of safe Petri nets are {\sf PSPACE}-complete~\cite{survey94}.
In~\cite{Haddad05}, it is proved that reachability and liveness are {\sf
PSPACE}-hard for safe $\piun$-nets and {\sf NP}-hard for  safe $\pideux$-nets.
The next theorem and its corollary improve on these results by showing that the
problem is not easier for
safe  $\pideux$-nets than for general safe Petri nets.

\begin{theo}
  The reachability problem for safe $\pideux$-nets is {\sf PSPACE}-complete. 
\end{theo}

\begin{proof}
  Our proof of {\sf PSPACE}-hardness is based on a reduction from the 
  {\sf QSAT} problem~\cite{Papadimitriou94}. 
  {\sf QSAT} consists in deciding whether the following formula is true\\
  \centerline{$\varphi \equiv \forall x_n \exists y_n \forall x_{n-1} \exists
  y_{n-1} \dots \forall x_{1} \exists y_{1} \psi$}\\ 
  where $\psi$ is a propositional formula over $\{x_1,y_1\ldots,x_n,y_n\}$ 
  in conjunctive normal form with at most three literals per clause.

  \smallskip \noindent 
  Observe that in order to check the truth of $\varphi$, one must check
  the truth of $\psi$ w.r.t. the $2^n$ interpretations of $x_1,\ldots,x_n$ while the 
  corresponding interpretation of any $y_i$ must only depend on the interpretation
  of $\{x_n,\ldots,x_i\}$.

  \smallskip \noindent {\bf Counters modelling.}
  First we design a $\pideux$-net
  $\mathcal N_{cnt}$ that ``counts'' from 0 to $2^k-1$. This net is defined by:
  \begin{itemize}[nolistsep]
    \item $P = \{ p_0,\dots, p_{k - 1}, q_0,\dots, q_{k - 1} \}$;
    \item $T = \left\{ t_0,\dots,t_{k - 1} \right\}$;
    \item For every $0\leq i < k$, ${}^\bullet t_i = p_i + \sum_{j<i} q_j$ and 
      $t_i^\bullet = q_i+ \sum_{j<i} p_j$;
    \item For every $0\leq i < k$, $m_0(p_i)=1$ and $m_0(q_i)=0$.
  \end{itemize}

  \begin{figure}[H]
    \begin{center}
      \includegraphics[scale=0.4]{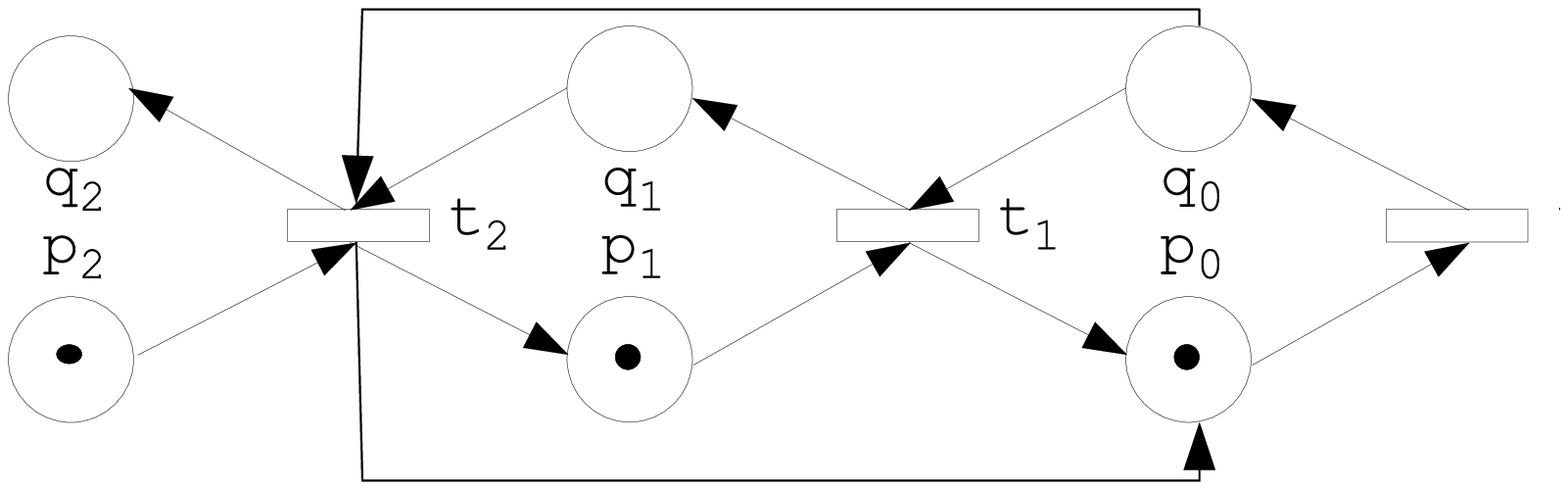}
    \end{center}  
    \caption{A $3$-bit counter (without the reverse transitions).}
    \label{fig:3bitcount}
  \end{figure}

  \noindent
  Observe that for every reachable marking $m$ and every index 
  $i$, we have $m(p_i) + m(q_i) = 1$. Therefore $m$ can be coded by the
  binary word
  $\w=\w_{k - 1}\dots\w_0$ in which $\w_i = m(q_i)$. The word $\w$
  is interpreted as the binary expansion of an integer between 0
  and $2^k-1$. We denote by $val(\w)$ the integer value associated
  with $w$. Consider $w \not\in \{0^k, 1^k\}$, 
  there are two
  markings reachable from $w$ which are $w+$ and $w-$ such that
  $val(w-)=val(w)-1$ and $val(w+)=val(w)+1$. 


  \smallskip \noindent
  The figure below represents the reachability graph of the $3$-bit counter.
  For a $k$-bit counter, the shortest firing sequence from $0^k$ to
  $1^k$ is $\sigma_k$ defined inductively by: $\sigma_1=t_0$ and $\sigma_{i+1}=\sigma_{i}t_i\sigma_{i}$.

  \begin{center}
    \includegraphics[scale=0.4]{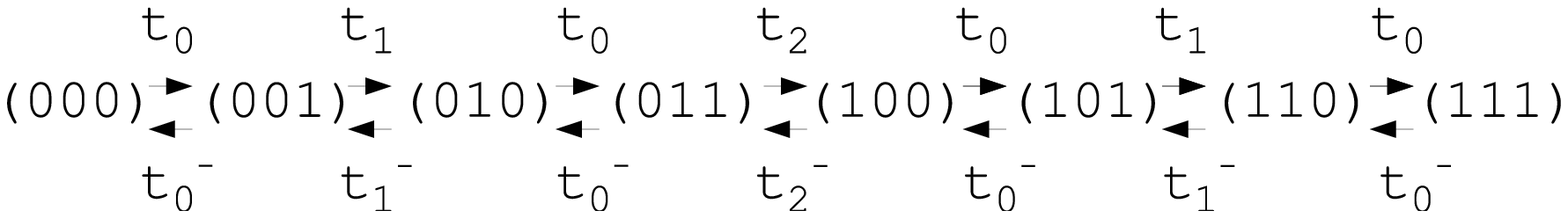}
  \end{center}

  \noindent
  For every complex $c \equiv p_i + \sum_{j<i} q_j$ (resp. $c \equiv q_i + \sum_{j<i} p_j$), a possible witness 
  is $wit(c)\equiv p_i + \sum_{j>i} 2^{j-i-1} p_j$ (resp. $wit(c) \equiv q_i + \sum_{j>i} 2^{j-i-1} q_j$).
  Thus this subnet has deficiency 0.

  \smallskip \noindent
  To manage transition firings between the update of counters, we duplicate
  the counter subnet
  and we synchronize the two subnets as indicated in the figure below. 	
  For a duplicated $k$-bit counter, the shortest firing sequence from the
  marking with the two counters set to $0^k$ and
  place $go$ marked to the marking with the two counters set to $1^k$ and
  place $go$ marked 
  is obtained by: 
  $\overline{\sigma}_1=\overline{t}_0$ and
  $\overline{\sigma}_{n+1}=\overline{\sigma}_{n}\overline{t}_n\overline{\sigma}_{n}$ 
  where $\overline{t}_i=t_it'_i$.

  \begin{center}
    \includegraphics[scale=0.35]{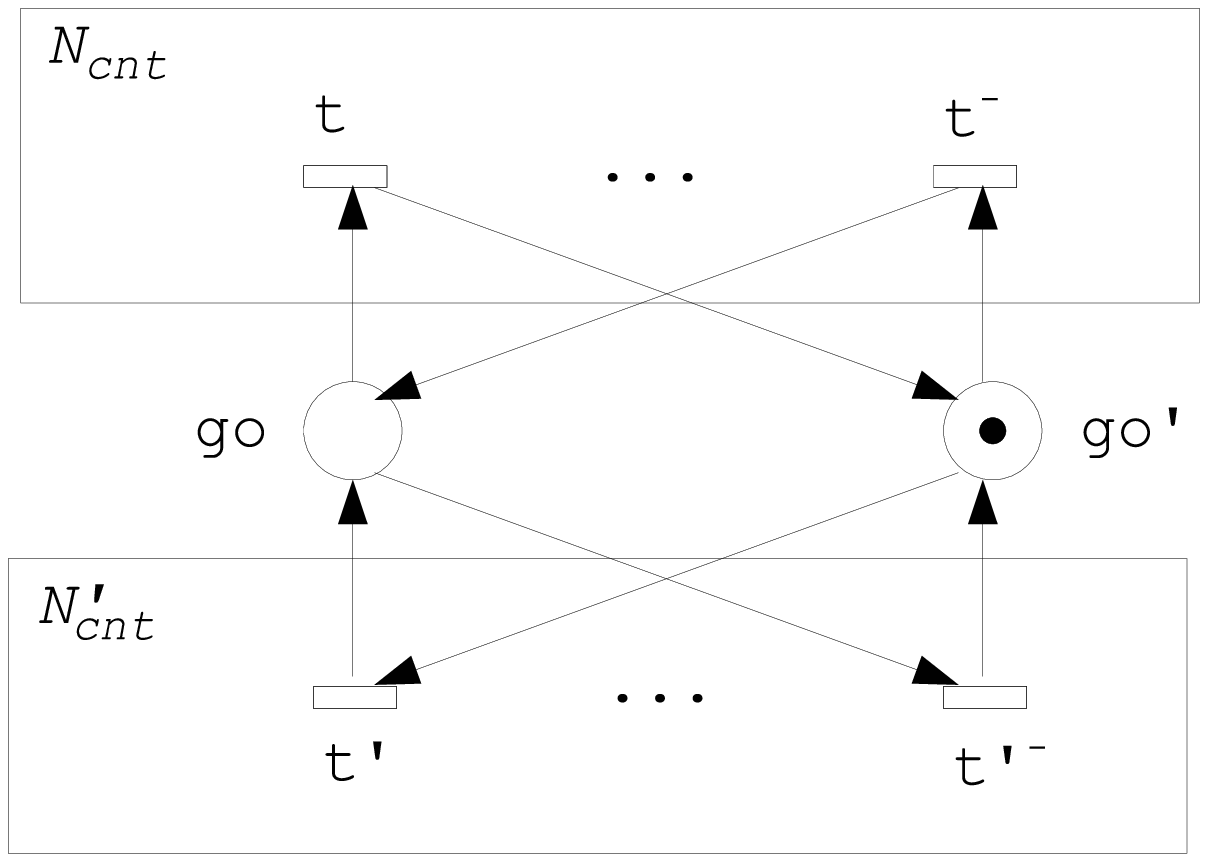}
  \end{center}

  \noindent
  This net has still deficiency 0 since the complexes are just
  enlarged by the places $go$
  or $go'$ and their witnesses remain the same. 



  \smallskip \noindent {\bf Variable modelling.}
  For reasons that will become clear later on, the two counter subnets
  contain $n+3$ bits indexed from $0$ to $n+2$. 
  The bits $1,\ldots,n$ of counter $cnt$ correspond to the value of variables $x_1,\ldots,x_n$.
  Managing the value of variables $y_1,\ldots,y_n$ is done as follows.
  For every variable $y_i$, we add the subnet described below on the
  left (observe that $s_i=r_i^{-}$) and
  modify the two counter subnets as described on the right.

  \begin{center}
    \includegraphics[scale=0.35]{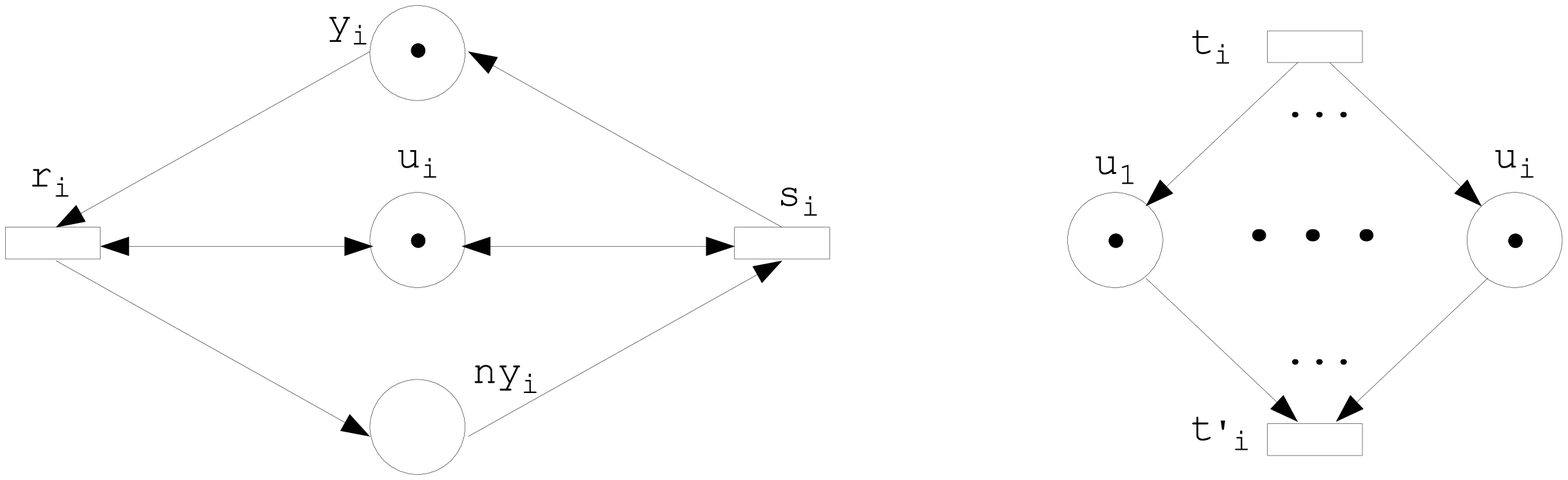}
  \end{center}

  \noindent
  When place $y_i$ (resp. $ny_i$) is marked, this corresponds to interpreting
  variable $y_i$ as {\bf true} (resp. {\bf false}). Changes of the interpretation
  are possible when place $u_i$ is marked. This is the role of the modification
  done on the counter subnet: between a firing of $t_i$ and $t'_i$ places
  $\{u_j\}_{j\leq i}$ are marked. With this construction, we get the
  expected behaviour: the interpretation of a variable $y_i$ can 
  only be modified when the interpretation of a variable $x_j$ with $j \geq i$
  is modified.
  The complexes of the counter subnet are enlarged with places $u_i$
  and their witnesses remain the same since places in the support of these
  witnesses are not modified by transitions $s_i$ and $r_i$. The new complex
  $y_i+u_i$ (resp. $ny_i +u_i$) has for witness $y_i$ (resp. $ny_i$). Thus
  the new net has still deficiency 0.  


  \smallskip \noindent {\bf Modelling the checking of the propositional formula.}
  We now describe the subnet associated with the checking of propositional formula
  $\psi \equiv \bigwedge_{j \leq m} C_j$ where we assume  w.l.o.g.: (1) that every clause 
  $C_j \equiv l_{j,1} \vee l_{j,2} \vee l_{j,3}$ has exactly three literals (i.e. variables
  or negated variables); and (2) that every variable or negated variable occurs at least
  in one clause. The left upper part of Figure~\ref{fig:clausej} shows the Petri net which describes clause $C_j$ of the
  formula $\psi$. Places $\ell_{j, k}$($k = 1, 2, 3$) 
  represent the literals while  places $n\ell_{j, k}$ represent the literal \emph{used as a proof of the clause}, 
  the place $mutex_j$ avoids to choose several proofs of the clause (and thus ensuring safeness), and finally place
  $success_j$ can be marked if and only if the evaluation of the clause
  yields true for the current interpretation and one of its true literal is used as a proof.

  \smallskip \noindent The complexes of this subnet are 
  $mutex_j+\ell_{j,k}$ (resp. $success_j+n\ell_{j,k}$) with
  witness $-n\ell_{j,k}$  (resp. $n\ell_{j,k}$). So the subnet has deficiency
  0.

  \begin{figure}[tbp]
    \begin{center}
      \includegraphics[scale=0.30]{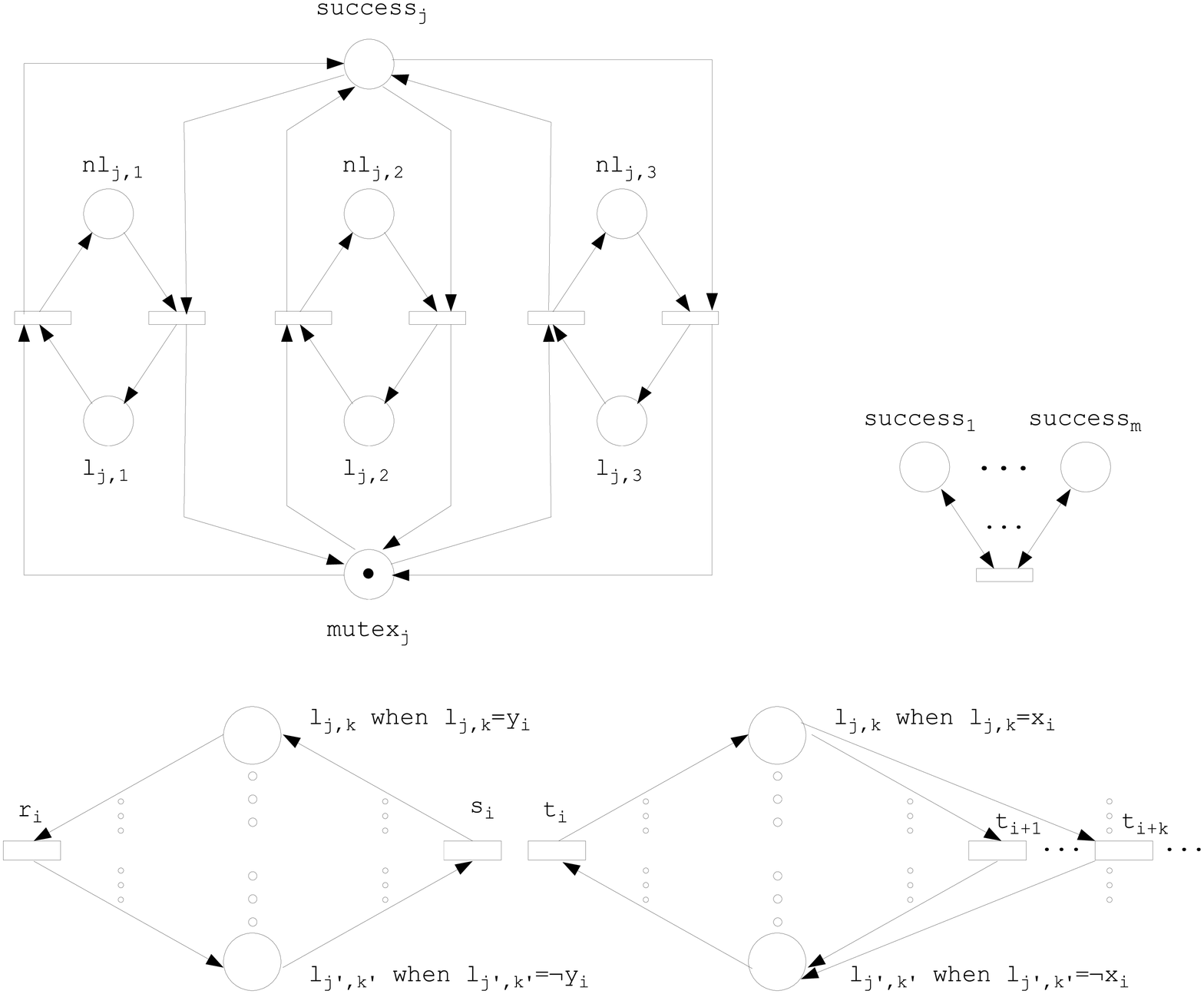}
    \end{center}
    \caption{Clause $C_j$ (left), synchronisation with $t_0$ (right) and with variables (below)}
    \label{fig:clausej}
  \end{figure}

  \smallskip \noindent 
  We now synchronise the clause subnets with the previous subnet in order to obtain
  the final net. Observe that in the previous subnet, transition $t_0$ (and $t'_0$) must occur
  after every interpretation change. This is in fact the role of bit 0 of the counter.
  Thus we constrain its firing by requiring
  the places $success_j$ to be marked as presented in the right upper part of Figure~\ref{fig:clausej}.
  Adding loops simply enlarges the complexes associated with $t_0$ and
  does not modify the incidence matrix. So the net has still deficiency 0.


  \smallskip \noindent 
  It remains to synchronise the value of the variables and the values of the literals
  where the variables occur either positively or negatively. This is done in two steps.
  First $\ell_{j, k}$ is initially marked if the interpretation  
  of the initial marking satisfies $\ell_{j, k}$. Then we synchronize the value
  changes as illustrated in the lower part of Figure~\ref{fig:clausej}. 
  Once again the complexes are enlarged and the witnesses
  are still valid since the places $\ell_{j, k}$ do not belong to the support of any witness.


  \smallskip \noindent  {\bf Choice of the initial and final marking for the net.}
  Let us develop a bit the sequence $\overline{\sigma}_{n+3}$ in the two counter subnet 
  in order to explain the choice of initial marking for this subnet:
  \[\overline{\sigma}_{n+3}=
  \overline{\sigma}_{n+1} t_{n+1} t'_{n+1} \overline{\sigma}_{n+1}
  t_{n+2}t'_{n+2} \overline{\sigma}_{n+1} t_{n+1}t'_{n+1}
  \overline{\sigma}_{n+1}\]

  \smallskip \noindent
  We want to check all the interpretations of $x_i$'s guessing the appropriate values of $y_i$'s (if they exist).
  We have already seen that changing from one interpretation to another one (i.e. a counter incrementation
  or decrementation) allows to perform the allowed updates of $y_i$. However given the initial interpretation of the $x_i$'s
  we need to make an initial guess of all the $y_i$'s. So our initial marking restricted to the counter subnet
  will correspond to the marking reached
  after $\overline{\sigma}_{n+1}t_{n+1}$, i.e. corresponding to $cnt=2^{n+1}$ (i.e. word $010\ldots 0$), 
  $cnt'=2^{n+1}-1$ (i.e. word $001\ldots 1$)
  with in addition places $go'$, $u_i$'s, $mutex_j$'s and $y_i$'s 1-marked; places  $\ell_{j, k}$ are marked 
  according to the initial marking of places $x_i$'s and $y_i$'s
  as explained before. All the other places are unmarked. This explains the role of bit $n+1$.

  \smallskip  \noindent
  Furthermore, if we have successfully checked all the interpretations of
  the $x_i$'s, the counters will have reached the value $2^{n+2}-1$
  (corresponding to a firing sequence obtained from  $t'_{n+1}\overline{\sigma}_{n+1}$
  with possible updates of $y_i$ during change of interpretations). However
  we do not know what is the final guess for the $y_i$'s. So firing
  transition $t_{n+2}$ allows to set the $y_i$'s in such a way that the final marking
  will correspond to $cnt= 2^{n+2}$ (i.e. word $10\ldots 0$), $cnt'= 2^{n+2}-1$ (i.e. word $01\ldots 1$) with in addition places $go'$,
  $u_i$'s  $mutex_j$'s and $y_i$'s 1-marked; places $\ell_{j, k}$  are marked accordingly. 
  All the other places are unmarked. This explains the role of bit $n+2$.

  \smallskip \noindent 
  By construction, the net reaches the final marking iff the formula is
  satisfied. Observe that the checking of clauses can be partially done
  concurrently with the change of interpretation. However as long as,
  in the net, a clause $C_j$ is ``certified'' by a literal $\ell_{j,k}$
  (i.e. marking place $success_j$ and unmarking place $\ell_{j,k}$) 
  the value of the variable associated
  with the literal cannot change, ensuring that when $t_0$ is fired, the marking of
  any place $success_j$ corresponds to the evaluation of clause $C_j$
  with the current interpretation.
\end{proof}

\begin{coro}
  The liveness problem for safe $\pideux$-nets is {\sf PSPACE}-complete. 
\end{coro}
\begin{proof}
  Observe that the transitions of the net of the previous proof
  are fireable at least once and so live by reversibility, implied by weak reversibility iff $\varphi$
  is true.
\end{proof}

Let us now consider general (non-safe) Petri nets. 
Reachability and coverability of symmetric nets is {\sf EXPSPACE}-complete~\cite{MM82}.
In~\cite{Haddad05}, it is proved that both problems are {\sf
EXPSPACE}-complete for WR nets
(which include symmetric Petri nets). 
The next proposition establishes the same result for the coverability
of $\pideux$-nets. 

\begin{prop}
  \label{prop:coverability}
  The coverability problem for $\pideux$-nets is {\sf EXPSPACE}-complete. 
\end{prop}

\begin{proof}
  \noindent
  Since we already know that coverability for $\piun$-nets belongs to
  {\sf EXPSPACE}~\cite{Haddad05}, it remains to prove that 
  coverability for $\pideux$-nets is {\sf EXPSPACE}-hard. 
  In order to establish this result, we slightly adapt the reduction
  given in~\cite{MM82}  of the termination problem
  for a three counter machine where the values of counters are
  bounded by $e_n \equiv 2^{2^n}$ with $n$ the size of 
  (a representation of) the machine. Thus we first
  depict the original reduction and then we
  describe our modifications and explain why the reduction
  is still valid.

  \smallskip \noindent
  For a uniform presentation of the proof we assume w.l.o.g. that
  the machine has four counters (these more powerful machines
  include the original ones). The key ingredient is 
  the concise management of counters
  and more precisely the zero test. Indeed one models
  a counter $c_i$ with $i \in \{1,2,3,4\}$ by two complementary
  places $A_{i,n}$ and $B_{i,n}$. When the counter has
  value $x$, place $A_{i,n}$ contains $x$ tokens and 
  place $B_{i,n}$ contains $e_n-x$ tokens. Testing (and decrementing) that the counter
  $c_i$ is greater than $0$ is done as usual by an arc with weight 1 starting
  from $A_{i,n}$. However this approach does not work for the zero test
  as it would require a (double) arc from $B_{i,n}$ with weight $e_n$ thus
  implying a net representation of size at least $2^n$ which would not be valid.

  \smallskip \noindent
  Thus the zero test is managed by an inductive construction (w.r.t. $n$)
  of  ``nested'' subnets  $\mathcal N_k$ 
  leading to a subnet (the union of these subnets) with size in $O(n)$. 
  Let us describe this
  construction. The main places are:  $B_{i,k}$ with $i \in \{1,2,3,4\}, 0\leq k \leq n$
  containing at most $e_k$ tokens and safe places $C_{i,k}$, $F_k$ and $S_k$.
  The inductive properties are the following ones:
  \begin{itemize}

    \item In subnet $\bigcup_ {l \leq k} \mathcal N_l$, starting from marking $S_k+C_{i,k}$
      one may reach marking $F_k+C_{i,k}+e_k B_{i,k}$. 

    \item Furthermore any marking reachable from 
      $S_k+C_{i,k}+\alpha_{i_1} B_{i_1,k} +\alpha_{i_2} B_{i_2,k}+\alpha_{i_3} B_{i_3,k}$ 
      ($\{i_1,i_2,i_3\}=\{1,2,3,4\}\setminus \{i\}$)
      with $S_k$
      or $F_k$ marked is either $S_k+C_{i,k}+\alpha_{i_1} B_{i_1,k} +\alpha_{i_2} B_{i_2,k}+\alpha_{i_3} B_{i_3,k}$
      or $F_k+C_{i,k}+e_k B_{i,k}+\alpha_{i_1} B_{i_1,k} +\alpha_{i_2} B_{i_2,k}+\alpha_{i_3} B_{i_3,k}$. 

  \end{itemize}

  \smallskip \noindent {\bf Basic case $k=0$.} This case is straightforward:
  $\mathcal N_0$ consists in four transitions when transition corresponding to $i$ is figured below.
  \begin{center}
    \includegraphics[scale=0.4]{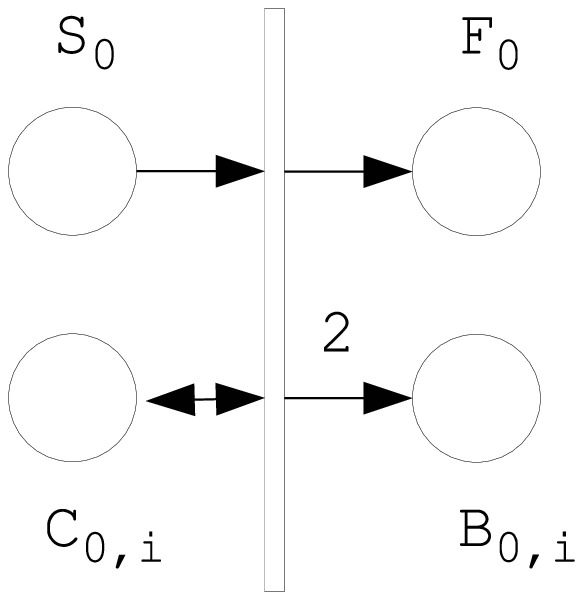}
  \end{center}

  \smallskip \noindent {\bf Inductive case.}
  Assume that the inductive properties holds for $k$. The net corresponding
  to $\mathcal N_{k+1}$ is described below with the following convention:
  $S$ corresponds to $S_{k+1}$ and $s$ corresponds to $S_{k}$. The
  same convention applies to all names. Furthermore for sake of readability
  we have duplicated some places in the figure.

  \begin{center}
    \includegraphics[scale=0.4]{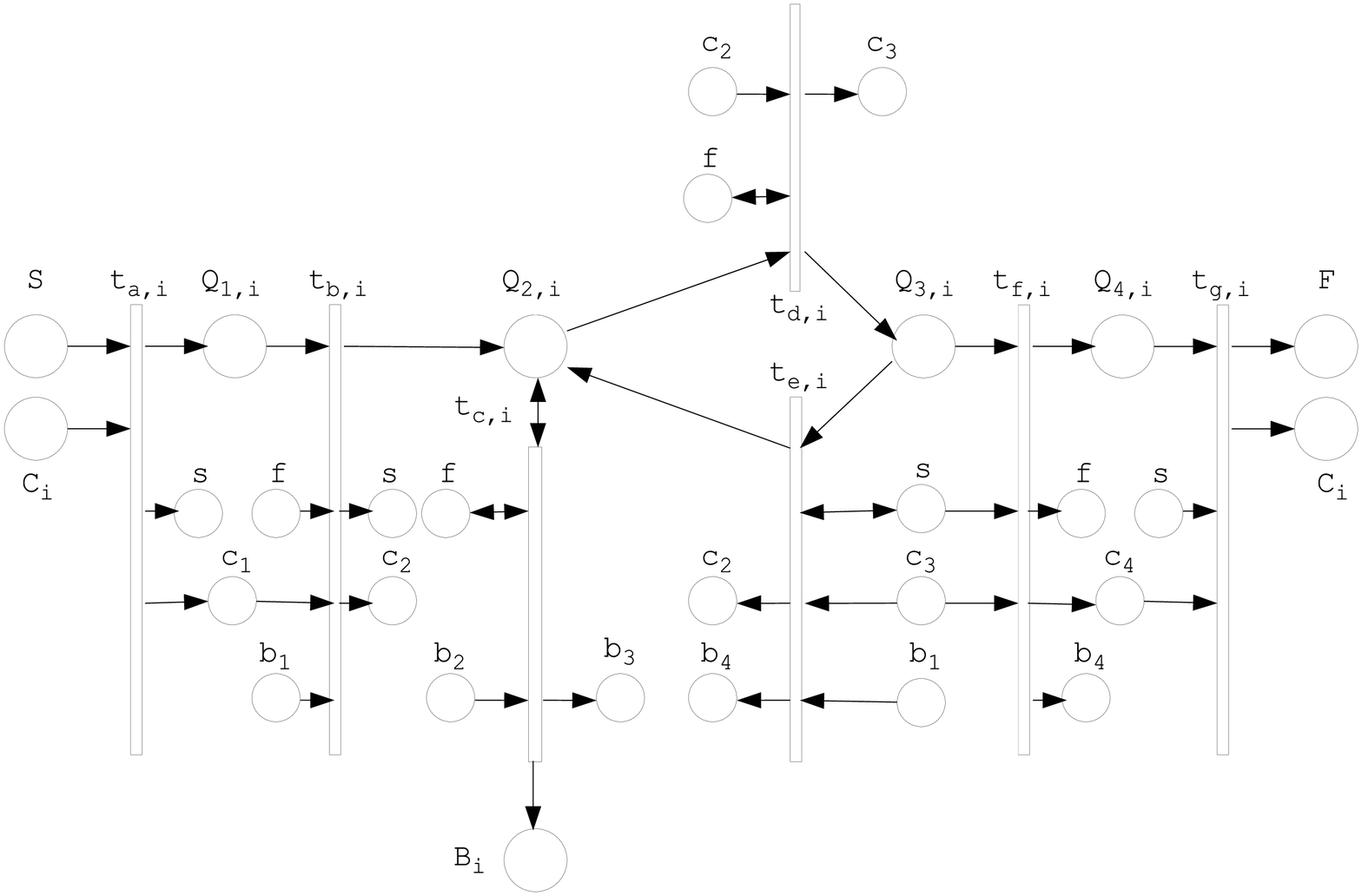}
  \end{center}

  \smallskip \noindent We first exhibit the firing sequence $\sigma_{i,k+1}$ from
  $S+C_i$ to $F+C_i+e_{k+1}B_i$:
  \begin{itemize}

    \item $S+C_i \xrightarrow{t_{a,i}} Q_{1,i}+s + c_1$

    \item $Q_{1,i}+s + c_1 \xrightarrow{\sigma_{1,k}} Q_{1,i}+ f + c_1+e_kb_1$
      using the inductive hypothesis

    \item $Q_{1,i}+ f + c_1+e_kb_1 \xrightarrow{t_{b,i}} Q_{2,i}+ s + c_2+(e_k-1)b_1$\\
      We now describe a firing sequence from\\ 
      $Q_{2,i}+ s + c_2+(e_k-j)b_1+(j-1)e_kB_i+(j-1)b_4$ to\\ 
      $Q_{2,i}+ s + c_2+(e_k-j-1)b_1+je_kB_i+jb_4$ for $1\leq j \leq e_k-1$
      \begin{footnotesize}
	\begin{itemize}

	  \item $Q_{2,i}+ s + c_2+(e_k-j)b_1+(j-1)B_i +(j-1)b_4$\\
	    $\xrightarrow{\sigma_{2,k}} Q_{2,i}+ f + c_2+(e_k-j)b_1+e_kb_2+(j-1)b_4$
	    using the inductive hypothesis

	  \item $Q_{2,i}+ f + c_2+(e_k-j)b_1+e_kb_2+(j-1)b_4$\\
	    $ \xrightarrow{(t_{c,i})^{e_k}} Q_{2,i}+ f + c_2+(e_k-j)b_1+e_kb_3+(j+1)e_kB_i+(j-1)b_4$

	  \item $Q_{2,i}+ f + c_2+(e_k-j)b_1+e_kb_3+(j+1)e_kB_i+(j-1)b_4$\\ 
	    $\xrightarrow{t_{d,i}} Q_{3,i}+ f + c_3+(e_k-j)b_1+e_kb_3+(j+1)e_kB_i+(j-1)b_4$

	  \item $Q_{3,i}+ f + c_3+(e_k-1)b_1+e_kb_3+(j+1)e_kB_i+(j-1)b_4$\\ 
	    $\xrightarrow{\sigma_{3,k}^-} Q_{3,i}+ s + c_3+(e_k-j)b_1+(j+1)e_kB_i+(j-1)b_4$

	  \item $Q_{3,i}+ s + c_3+(e_k-j)b_1+(j+1)e_kB_i +(j-1)b_4$\\
	    $\xrightarrow{t_{e,i}} Q_{2,i}+ s + c_2+(e_k-j-1)b_1+(j+1)e_kB_i+jb_4$

	\end{itemize}
      \end{footnotesize}

    \item After the previous iterations, when reaching $Q_{2,i}+ s + c_2+(e_k-1)e_kB_i+(e_k-1)b_4$,
      we perform all the steps of the iteration except the last one reaching
      $Q_{3,i}+ s + c_3+(e_k)^2B_i+(e_k-1)b_4=Q_{3,i}+ s + c_3+e_{k+1}B_i+(e_k-1)b_4$.        

    \item $Q_{3,i}+ s + c_3+e_{k+1}B_i+(e_k-1)b_4 \xrightarrow{t_{f,i}} Q_{4,i}+ f + c_4+e_{k+1}B_i+e_kb_4$

    \item $Q_{4,i}+ f + c_4+e_{k+1}B_i+e_kb_4 \xrightarrow{\sigma_{4,k}^-} Q_{4,i}+ s + c_4 +e_{k+1}B_i$

    \item $Q_{4,i}+ s + c_4 +e_{k+1}B_i \xrightarrow{t_{g,i}} F+ C_i +e_{k+1}B_i$

  \end{itemize}

  \smallskip \noindent
  Let us now prove that any marking reachable from 
  $S_k+C_{i,k}+\alpha_{i_1} B_{i_1,k} +\alpha_{i_2} B_{i_2,k}+\alpha_{i_3} B_{i_3,k}$ 
  ($\{i_1,i_2,i_3\}=\{1,2,3,4\}\setminus \{i\}$)
  with $S_k$
  or $F_k$ marked is either $S_k+C_{i,k}+\alpha_{i_1} B_{i_1,k} +\alpha_{i_2} B_{i_2,k}+\alpha_{i_3} B_{i_3,k}$
  or $F_k+C_{i,k}+e_k B_{i,k}+\alpha_{i_1} B_{i_1,k} +\alpha_{i_2} B_{i_2,k}+\alpha_{i_3} B_{i_3,k}$. 
  We first observe on the net above that the tokens contained in a place $B_{j,k}$ are frozen
  except when place $C_{j,k}$ is marked. Thus w.l.o.g. we assume that $\alpha_1= \alpha_2=\alpha_3=0$.

  \smallskip \noindent
  So it remains to show that when deviating from the exhibited sequence one cannot reach a marking
  with $S_k$ marked different from the initial marking or a marking with $F_k$ marked different from the final
  marking. This is proven by a case study (see~\cite{MM82}). Here we just handle one case since all
  cases are similar. When reaching marking $Q_{2,i}+ f + c_2+(e-j)b_1+e_kb_3+(j+1)e_kB_i+(j-1)b_4$
  with $0\leq e<e_k$, one can fire transition $t_{d,i}$ reaching marking 
  $Q_{3,i}+ f + c_3+(e-j)b_1+eb_3+(j+1)e_kB_i+(j-1)b_4$. From this marking due to the inductive 
  hypothesis, it is not possible to mark place $s$ in subnet $\mathcal N_{k-1}$. Thus
  transitions $t_{e,i}$ and $t_{f,i}$ are not fireable. So the only possible way to  ``progress'' in 
  $\mathcal N_k$ consists to fire the reverse transition $t_{d,i}^-$ coming back to the marking
  $Q_{2,i}+ f + c_2+(e-j)b_1+e_kb_3+(j+1)e_kB_i+(j-1)b_4$.

  \smallskip \noindent
  The subnet below describes the initial behaviour of the simulating net consisting in filling
  places $B_{i,n}$ (with $i \in \{1,2,3,4\}$) with $e_n$ tokens and putting a token in $q_0$
  the place corresponding to the initial state of the counter machine.

  \begin{center}
    \includegraphics[scale=0.4]{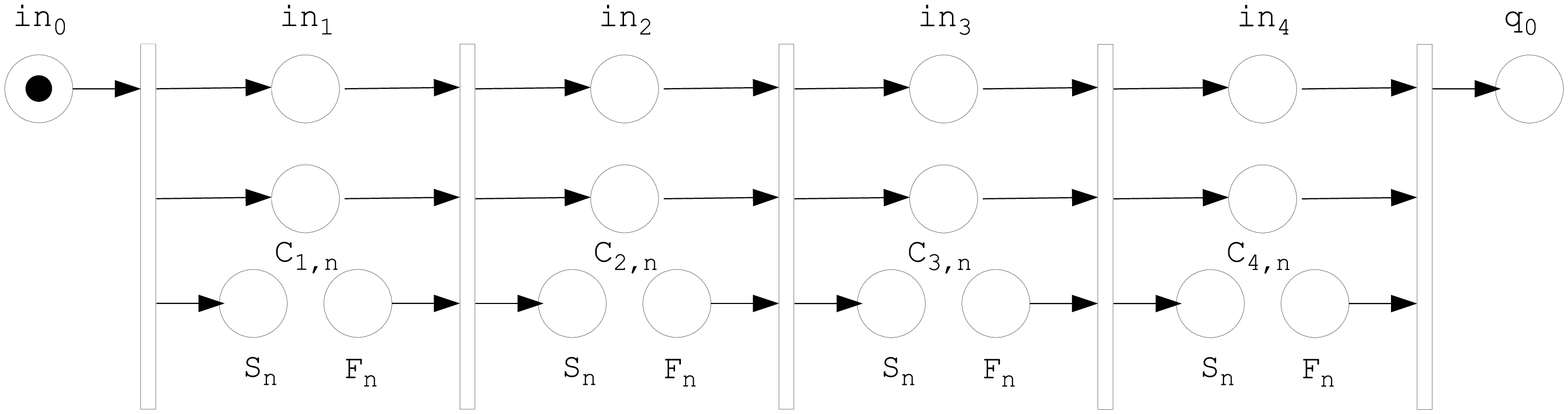}
  \end{center}

  \smallskip \noindent
  The simulation of an instruction $q: {\sf if} c_i> 0 {\sf ~then~ } c_i\!-\!-; {\sf ~goto~ } q' {\sf ~else~goto~ } q''$
  is now simply performed by the following subnet. The validity of the zero test is ensured by
  the assertions about the subnet  $\bigcup_ {l \leq n} \mathcal N_l$. Furthermore
  it can be proved that reverse transitions of the ones simulating transitions
  cannot help to mark place $q_f$ where $q_f$ is the final state of the counter
  machine (see~\cite{MM82} or proposition 12 in~\cite{Haddad05} for a simple proof of this claim).

  \begin{center}
    \includegraphics[scale=0.4]{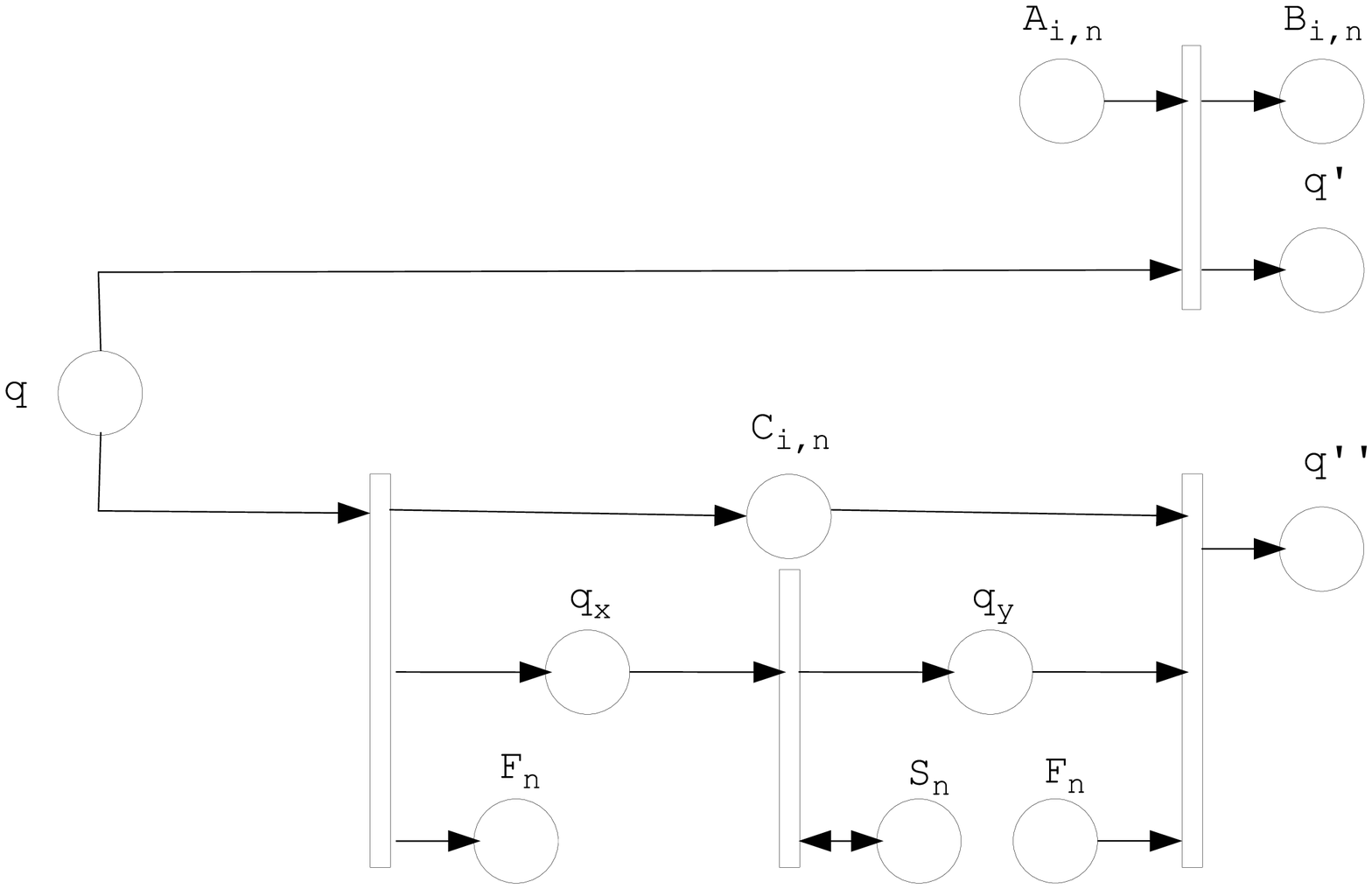}
  \end{center}

  \smallskip \noindent
  We are now ready to explain the modifications that we bring to the simulating net.
  For every pair of transitions $t$ and $t^-$, we add a place $p_t$ input of
  one of the transitions and output of the other. Thus by construction $p_t$
  and $-p_t$ are witnesses for $t$ and $t^-$. More precisely
  $p_t$ is the witness of the
  transition for which it is an output and $-p_t$ is the output of the other
  transition. Let us examine how these additional places modify the behaviour
  of the net. Since there is no new transition, firing sequences of the enlarged net
  are firing sequences of the original one. Thus we only have to care whether
  the simulation firing sequence is still a firing sequence. 

  \smallskip \noindent
  For the transitions not belonging to the subnet  $\bigcup_ {l \leq n} \mathcal N_l$, $p_t$
  is an output of $t$. As the reversed transitions of these transitions do not occur in
  the simulating sequence, such places cannot disable a transition in the simulating
  sequence. We now observe that the sequences $\sigma_{i,n}$ and $\sigma_{i,n}^-$
  alternate in the simulating sequence,
  always starting by $\sigma_{i,n}$. Thus in subnet $\mathcal N_n$, place $p_t$ is the output of the transition $t$.
  Now observe that in sequence $\sigma_{i,n}$ there is an occurrence of sequence $\sigma_{1,n-1}$,
  $e_{n-1}$ occurrences of $\sigma_{2,n-1}$ followed by $\sigma_{3,n-1}^-$ and then an occurrence
  of $\sigma_{4,n-1}$. Thus in subnet $\mathcal N_{n-1}$, place $p_t$ is the output of the transition $t$ (resp. $t^-$)
  when $t$ is $t_{u,i}$ with $u \in \{a,b,c,d,e,f,g\}$ and $i \in \{1,2\}$
  (resp. $i \in \{3,4\}$)using notations of the figure. The same pattern of
  occurrences
  also happens at lower levels. So more generally, in subnet $\mathcal N_{k}$
  with $k<n$, place $p_t$ is the output of the transition $t$ (resp. $t^-$)
  when $t$ is $t_{u,i}$ with $u \in \{a,b,c,d,e,f,g\}$ and $i \in \{1,2\}$
  (resp. $i \in \{3,4\}$).
  \smallskip \noindent
  With this choice, the simulating sequence is still a firing sequence in the enlarged net and
  the marking to be covered is $q_f$.
\end{proof}

\medskip

The complexity of reachability for $\pideux$-nets remains an open issue (indeed
the proof of {\sf EXPSPACE}-hardness does not work for reachability).

\section{The subclass of \pitrois-nets}
\label{section:normalization}
In this section, we introduce $\pitrois$-nets, a
subclass of product-form Petri nets for which the normalising constant can be
efficiently computed. 
The first subsection defines the subclass; the second one studies its
structural properties and the third one is devoted to the computation
of the normalising constant.


\subsection{Definition and properties}\label{subsec:construction}

\begin{defi}[Ordered \piun-net] 
Consider an integer $n\geq 2$. An $n$-level
  {\em ordered \piun-net} is a \piun-net $\cN = (P, T, W^-, W^+)$ such that:

  \smallskip

  \begin{enumerate}[nolistsep]
    \item $P = \bigsqcup\limits_{1 \le i \le n} P_i\,,~ T =  \bigsqcup\limits_{1 \le i \le n}
      T_i$ and $P_i \neq \emptyset$ for all $1 \le i \le n$,
    \item $\cM_i = (P_i, T_i, W^-_{|P_i \times T_i},  W^+_{|P_i \times T_i})$ is a
      strongly connected state machine,
    \item $\forall 1 \le i \le n\,, \forall t \in T_i\,, \forall p \in P$,
      ${}^\bullet t(p) > 0$ implies $p \in P_i$ or $p \in P_{i - 1}$ ($P_0 = \emptyset$),
    \item $\forall 2 \le i \le n\,, \exists t \in T_i\,, \exists p \in
      P_{i - 1}$ s.t. ${}^\bullet t(p) > 0$,
    \item $\forall 1 \le i \le n\,, \forall t, t' \in T_i$, $({}^\bullet t \cap
      {}^\bullet t')  \cap P_i \neq \emptyset$ implies ${}^\bullet t = {}^\bullet t'$.
  \end{enumerate}

  \smallskip
  
  We call $\cM_i$ the level $i$ state machine. The elements of $P_i$ (resp.
  $T_i$) are level $i$ places (resp. transitions).
  The complexes ${}^\bullet t$ with 
  $t \in T_i$ are level $i$ complexes.
\end{defi}

By weak reversibility, 
the constraints 3, 4, and 5 also apply to the output bags $t^\bullet$.
An {\em ordered \piun-net} is a sequence of strongly connected state machines.
Connections can only be made between a level
$i$ transition and a level $(i - 1)$ place (points 1, 2, 3). By
construction, an {\em 
ordered \piun-net} is connected (point 4). For $i > 1$, each level $i$ place
belongs to one and only one level $i$ complex (point 5). 
An example of ordered  \piun-net can be found on  figure~\ref{fig:pi3}.

\begin{lemm}\label{lem:Nreactiongraph}
  The reaction net of $\cN$
  is isomorphic to the disjoint union of
  state machines $\cM_i$.
  Consequently, a $T$-semi-flow of $\cM_i$ is also a    
  $T$-semi-flow of $\cN$.
  If a transition of $T_i$ is enabled by a reachable marking then every
  transition of $T_i$ is live.
\end{lemm}

\begin{proof}
  Consider the mapping $f$ which maps each complex
  $t^\bullet$, $t \in T_i$, to $p$ the output place of $t$ in $P_i$. By
  construction of ordered \piun-nets, $f$ is a bijection from $\cC$ to
  $P$. Moreover, each arc $c_1 \rightarrow c_2$ of the reaction graph
  corresponds to the transition $t = f(c_1)^\bullet = { }^\bullet{f(c_2)}$.
  This proves the first point of the lemma.\\
  To prove the second point, recall that
  for a state machine, the $T$-semi-flows correspond to circuits
  of the Petri net graph. From this and from the first point, a $T$-semi-flow
  of $\cM_i$ defines a circuit of the reaction graph of $\cN$, which yields a
  $T$-semi-flow of $\cN$.\\
  The set of transitions $T_i$ is the set of transitions occurring in a component
  of the reaction graph. The third point follows.
\end{proof}

An ordered \piun-net may be interpreted as a multi-level system. The transitions
represent jobs or events while the tokens in the places represent
resources or constraints. A level $i$ job requires resources from level
$(i - 1)$ and relocates these resources upon completion. On the contrary, events
occurring in level $(i - 1)$ may make some resources unavailable, hence
interrupting activities in level $i$. The dependency of an activity on the next
level is measured by {\em potentials}, defined as follows.

\begin{defi}[Interface, potential]
  A place $p \in P_i$, $1 \leq i \leq n - 1$, is an {\em interface place} if $p \in
  t^{\bullet}$ for some $t \in T_{i + 1}$. For 
a place $p \in P_i$, $2 \leq i \leq n$,
and a place $q \in P_{i - 1}$, set:
  \[
  pot(p, q) = 
  \begin{cases}
    t^{\bullet}(q) & \text{ if } p \text{ and } q \text{ have a common input
    transition } t \in T_i\\
    0 & \text{otherwise.}
  \end{cases}
  \]
  The {\em potential} of a place $p \in P_{i}$, $2 \leq i$, is
  defined by:
  \[
  pot(p) = \sum_{q \in P_{i - 1}} pot(p, q)\,.
  \]
  By convention, $pot(p) = 0$ for all $p \in P_1\,.$
\end{defi}

By the definition of ordered \piun-nets, the quantity $t^{\bullet}(q)$ does not depend 
on the choice of
$t$, so the potential is well-defined. Indeed,
by weak reversibility, the constraint 5 also applies to the output bags $t^\bullet$.

\smallskip

\noindent{\bf Example.} The Petri net in Figure \ref{fig:pi3} is a 3-level
ordered \piun-net. The potentials are written in parentheses. To
keep the figure readable, the arcs between the place $p_1$ and the level 2
transitions are not shown.

\begin{figure}[H]
\centering
\begin{mfpic}{-140}{140}{-40}{40}
\tlabelsep{2pt}

\circle{(-60, 20), 4}
\tlabel[cl](-56, 20){{\scriptsize $p_3(2)$}}
\circle{(-60, -20), 4}
\tlabel[cl](-56, -20){{\scriptsize $q_3(1)$}}
\circle{(-90, 0), 4}
\tlabel[cr](-94, 0){{\scriptsize $r_3(0)$}}

\circle{(0, 20), 4}
\tlabel[cl](4, 20){{\scriptsize $p_2(2)$}}
\circle{(0, -20), 4}
\tlabel[cl](4, -20){{\scriptsize $q_2(2)$}}
\circle{(30, 0), 4}
\tlabel[cl](34, 0){{\scriptsize $r_2(1)$}}

\circle{(75, 0), 4}
\tlabel[cl](79, 0){{\scriptsize $p_1$}}

\shiftpath{(-60, 40)}\rect{(-5, -1), (5, 1)}
\shiftpath{(-60, 0)}\rect{(-5, -1), (5, 1)}
\shiftpath{(-60, -40)}\rect{(-5, -1), (5, 1)}

\shiftpath{(0, 40)}\rect{(-5, -1), (5, 1)}
\shiftpath{(0, 0)}\rect{(-5, -1), (5, 1)}
\shiftpath{(0, -40)}\rect{(-5, -1), (5, 1)}

\arrow\polyline{(-60, 38), (-60, 25)}
\arrow\shiftpath{(0, -23)}\polyline{(-60, 38), (-60, 25)}
\arrow\shiftpath{(0, -40)}\polyline{(-60, 38), (-60, 25)}
\arrow\shiftpath{(0, -63)}\polyline{(-60, 38), (-60, 25)}

\arrow\arc[s]{(-65, -40), (-90, -5), -60}
\arrow\arc[s]{(-90, 5), (-65, 40), -60}

\arrow\shiftpath{(60, 0)}\polyline{(-60, 38), (-60, 25)}
\arrow\shiftpath{(60, 0)}\shiftpath{(0, -23)}\polyline{(-60, 38), (-60, 25)}
\arrow\shiftpath{(60, 0)}\shiftpath{(0, -40)}\polyline{(-60, 38), (-60, 25)}
\arrow\shiftpath{(60, 0)}\shiftpath{(0, -63)}\polyline{(-60, 38), (-60, 25)}

\arrow\arc[s]{(30, 5), (5, 40), 60}
\arrow\arc[s]{(5, -40), (30, -5), 60}

\arrow\polyline{(-55, 38), (-5, 21)}
\tlabel[bc](-40, 34){{\scriptsize $2$}}
\arrow\polyline{(-5, 19), (-55, 2)}
\tlabel[tc](-20, 11){{\scriptsize $2$}}

\arrow\shiftpath{(0, -40)}\polyline{(-55, 38), (-5, 21)}
\arrow\shiftpath{(0, -40)}\polyline{(-5, 19), (-55, 2)}


\dashed\polyline{(-30, 45), (-30, -45)}
\dashed\polyline{(60, 45), (60, -45)}

\tlabel[bc](-90, 40){{\scriptsize level 3}}
\tlabel[bc](30, 40){{\scriptsize level 2}}
\tlabel[bc](75, 40){{\scriptsize level 1}}

\point[1.5pt]{(-60, 20), (-60, -20), (-90, 0), (73.5, 1), (76.5, 1), (73.5, -1),
(76.5, -1)}

\end{mfpic}
\caption{Ordered \piun-net.}
\label{fig:pi3}
\end{figure}

\begin{defi}[Marking witness]\label{def:effectivemarking}
  The {\em marking witness} of a marking $m$, denoted by $\widetilde m$, is
  defined as follows. For all $i \le n$ and $p \in P_i$,
\begin{small}
  \begin{equation}
    \label{eq:effectivemarking}
    \widetilde m(p) = m(p) +
    \sum_{j = 1}^{n - i}\bigl((-1)^j \sum_{\substack{r_1 \in P_{i + 1}\\\dots\\
    r_j \in P_{i + j}}} m(r_j)\bigl(\prod_{k = 1}^{j - 1}pot(r_{k + 1}, r_k)\bigr)
    pot(r_1, p)\bigr)\,.
  \end{equation}
\end{small}
\end{defi}

\noindent{\bf Remark.} Note that a marking witness is not necessarily
non-negative. It can be showed by induction that:\\
\begin{small}
\centerline{$\forall p \in P_n\,, \widetilde m(p) = m(p)$
and $\forall p \in P_i\,, i < n\,, \widetilde m(p) = m(p) -
    \sum_{r \in P_{i + 1}} \widetilde m(r)pot(r, p)$}
\end{small}

\begin{lemm}\label{lem:behaviorofmtilde}
  Let $m$, $m'$ be two vectors 
  such that $m' = m + W(t)$ for some $t \in T_i\ (1 \leq i
  \leq n)$. Let $p_1$ and $p_2$ denote the
  input place and the output place of $t$ in $P_i$, respectively. Then for every
  place $p$:
  \begin{equation}\label{eq:mtilde}
  \widetilde m'(p) =
      \widetilde m(p) - 1  \mbox{ if $p$ is $p_1$,} \ \ \
    \widetilde m(p) + 1  \mbox{ if $p$ is $p_2$,} \ \ \
    \widetilde m(p)  \mbox{ otherwise.}
 \end{equation}
\end{lemm}

\begin{proof}
  Since $m$ and $m'$ have the same restriction on $\cup_{j > i} P_j$, we have
  $\widetilde m'(p) = \widetilde m(p)\ \forall p \in \left(\cup_{j \geq i}
  P_j\right) \setminus \{p_1, p_2\}$.
  It follows that $\widetilde m'(p_1) - \widetilde m(p_1) = m'(p_1) - m(p_1) =
  -1$ and $\widetilde m'(p_2) - \widetilde m(p_2) = m'(p_2) - m(p_2) = 1$.\\
  For $p \in P_{i - 1} \cap t^\bullet$, we have $m'(p) - m(p) = pot(p_2,
  p) - pot(p_1, p)$, hence 
  \begin{align*}
    \widetilde m'(p) - \widetilde m(p) &= m'(p) - m(p) - \left[(\widetilde m'(p_1) -
    \widetilde m(p_1))pot(p_1, p)\right.\\
    &  \left.+ (\widetilde m'(p_2) - \widetilde
    m(p_2))pot(p_2, p)\right]\\
    &= 0
  \end{align*}
  Similarly,
  $\widetilde m'(p) - \widetilde m(p) = 0$ for $p \in P_{i - 1} \cap
  {}^\bullet{t}$.\\
  For all other places, $m'(p) = m(p)$ and $\widetilde m'(r) =
  \widetilde m(r)\ \forall r \mbox{ s.t. } pot(r, p) \neq 0$, thus
  $\widetilde m'(p) = \widetilde m(p)$.
\end{proof}

The above lemma applies in particular when $m$ and $m'$ are
markings such that $m\stackrel{t}{\longrightarrow} m'$. 
Equations~\eref{eq:mtilde} look like the equations for
witnesses. Since each level $i$ complex contains exactly one
level $i$ place, one guesses that every complex admits a witness,
i.e. that 
$\cN$ is a \pideux-net. This is confirmed by the next proposition.

\begin{prop}\label{prop:witnessmatrix}
  Let $B$ denote the $P \times P$ integer matrix of the linear transformation
  $m \mapsto \widetilde m$ defined by \eref{eq:effectivemarking}.
  For $p\in P_i$, the line vector $B(p)$ is a witness for the $i$-level complex
  containing $p$.  In particular, $\cN$ is a \pideux-net. 
\end{prop}

\begin{proof}
  Denote by $A \in \mathbb Z(\cC \times T)$ the incidence matrix of
  the reaction graph. From Lemma
  \ref{lem:behaviorofmtilde}, we have:
  \[
  m \xrightarrow{t} m' \implies \widetilde m' - \widetilde m = A(t)\,.
  \]

\noindent  
We have to show that 
  $BW(t) = A(t)~\forall t \in T$. Indeed, let $m$ and
  $m'$ be two markings such that $m \xrightarrow{t} m'$, we have: 
  $BW(t) = B(m' - m) = \widetilde m' - \widetilde m = A(t)$.
\end{proof}

Lemma~\ref{lem:behaviorofmtilde} allows to derive relevant  
S-semi-flows of $\cN$ and S-invariants.

\begin{coro}\label{cor:sinvariants}
  Let $m_0$ be the initial marking of $\cN$. We have:\\
 \centerline{$
\forall m \in \mathcal R (m_0), \quad \forall i \in \{1,\dots , n\}, \qquad 
\widetilde m(P_i) = \widetilde m_0(P_i)
 $}\\
  More generally, for all $i$, the vector 
  $v_i = \sum_{p \in P_i} B(p)$ is a S-semi-flow of $\cN$. 
\end{coro}
Using this corollary, it can be shown that an ordered \piun-net is bounded.

\smallskip \noindent{\bf Example.} 
Consider the ordered \piun-net in Figure \ref{fig:pi3} with the initial marking
$m_0 = p_3 + q_3 + r_3 + 4 p_1$. The marking witness of $m_0$ is 
$\widetilde{m_0} = p_3 + q_3 + r_3 - 2 p_2 - q_2 + 10 p_1$.
Any reachable marking $m$ satisfies the invariants:\\
\centerline{$m(P_3) = 3$}\\
\centerline{$m(P_2) - 2 m(p_3) - m(q_3) = -3$}\\
\centerline{$m(p_1) - 2 m(p_2) - 2 m(q_2) - m(r_2) + 4 m(p_3) + 2 m(q_3) = 10$}\\
We shown that $\{v_i, 1 \le i \le n\}$ is
a basis of the S-semi-flows of $\cN$. 

\begin{prop}\label{prop:sinv-uniqueness}
  Let $v$ be an S-semi-flow of $\cN$, i.e. $v . W = 
  0$. There exist unique rational numbers $a_1,\dots,a_n$ such that
  $v = \sum_{i = 1}^n a_i v_i$.
\end{prop}

\begin{proof}
  The matrix $B$ is a $P \times P$ unit lower triangular matrix, 
  so it is invertible.

\noindent
  We have: 
  \[
  v . W = 0 \implies (v .
  B^{-1})(BW) = 0 \implies (v . B^{-1}) A = 0\,,
  \]
  hence
  $v . B^{-1}$ is an S-semi-flow of the disjoint union of the state
  machines $\cM_i$. But since a state machine's only S-semi-flows are
  $a(1,\dots,1)$, $a \in \mathbb Q$, there exist rational numbers
  $a_1,\dots,a_k$ such that 
  \begin{equation}\label{eq:xB-1}
    v . B^{-1} = \sum_{i = 1}^n
    a_i w_i\,, 
  \end{equation}
  where $w_i \in \mathbb Q^{P}$ are
  defined by $w_i(p) = \mathbb 1_{P_i}(p)$.

\noindent  Right-multiplying both sides of (\ref{eq:xB-1}) by $B$, we get $v = \sum_{i = 1}^n
  a_i v_i$.

\noindent
The independence of the set $\{v_i\,, 1 \le i \le n\}$ follows from the fact that the
vectors $v_iB^{-1}$ have non-empty disjoint supports.
\end{proof}

\smallskip

We now consider only
ordered \piun-nets in which the interface
places in $P_i$ have maximal potential among the places of $P_i$. 
From the technical point of view,
this assumption is crucial for the reachability set analysis presented later. 
From the modelling point of view, it is a reasonable
restriction. Consider the
multi-level model, the assumption means that during the executions of level
$i$ jobs, the level $(i - 1)$ is idle, therefore the amount of available
resource is maximal.

\begin{defi}[\pitrois-net]
  An ordered \piun-net $\cN$ is a {\pitrois-net} if:\\
  \centerline{$
  \forall i, \forall p \in P_i\,: p \in {}^\bullet T_{i + 1} 
  \implies
  pot(p) = max \{pot(q), q \in P_i\}\,.
  $}\\
\end{defi}

\subsection{The reachability set}\label{subsec:reachabilityset}

From now on, $\cN$ is a $n$-level \pitrois-net with
$\cM_1, \dots, \cM_n$ being its state machines.

\begin{defi}[Minimal marked potential]
  Consider $i \in \{2,\dots , n\}$. The {\em level $i$ minimal potential marked by $m$} is:
  \[ 
  \varphi_i(m) = \begin{cases}
    \mbox{max}\{pot(p), p \in P_i\} & \mbox{ if } m(P_i) = 0\,,\\
    \mbox{min}\{pot(p), p \in P_i, m(p) > 0\} & \mbox{ if } m(P_i) > 0\,.
  \end{cases}
  \]
\end{defi}

The next lemma gives a necessary condition for reachability.

\begin{lemm}
  \label{lem:preserve}
  If $\varphi_i(m) \le m(P_{i - 1})$ then $\varphi_i(m') \le m'(P_{i -
  1})$ for all $m' \in \cR(m)$.
\end{lemm}

\begin{proof}
  W.l.o.g., assume that $m \xrightarrow{t} m'$.

  \noindent  First, suppose that $t \notin T_i$. If $t \notin T_{i + 1}$
  then firing $t$ does not modify the marking on $P_i$, so $\varphi_i(m') =
  \varphi_i(m)$. If $t \in T_{i + 1}$, firing $t$ either leaves the marking of
  $P_i$ unchanged or moves tokens between places of maximal potential in
  $P_i$; in both cases $\varphi_i(m') = \varphi_i(m)$.
  Since  $t \notin T_i$, $m'(P_{i - 1}) = m(P_{i - 1})$. So
  $\varphi_i(m') \le m'(P_{i - 1})$ if $t \notin T_i$.

  \noindent  Now consider $t \in T_i$, let $p$ and $q$ be the input and output places of
  $t$ in $P_i$. We
  have $\varphi_i(m') \le pot(q) \le m(P_{i - 1}) - pot(p) + pot(q) =
  m'(P_{i - 1})$.
\end{proof}

We now define the partial liveness and partial reachability.

\begin{defi}[$i$-reachability set, $i$-liveness]
  Let $m$ be a marking.
  The $i$-reachability set of $m$, denoted by $\cR_i(m)$, 
  is the set of all markings reachable from
  $m$ by a firing sequence consisting of transitions in $\bigcup_{1
  \le j \le i} T_j$.
  We say that $m$ is {\em $i$-live} if for any transitions $t$ in 
  $\bigcup_{1 \le j \le i} T_j$, there exists a marking in $\cR_i(m)$ which
  enables $t$.
  By convention, $\cR_0(m) = \{m\}$ and every marking is $0$-live.
\end{defi}



The $i$-live markings are characterised by the following proposition.

\begin{prop}\label{prop:ilivecondition}
  A marking $m$ is $i$-live if and only if it satisfies the following
  inequalities, called the {\em $i$-condition}:
  \begin{equation}\label{eq:icondition}
    m(P_i) > 0 \wedge
    \forall 2 \le j \le i\,:~m(P_{j - 1}) \ge \varphi_j(m)
  \end{equation}
  If $m$ satisfies the $i$-condition then
  for every $p, q \in P_i$ such that $p \neq q$, $m(p) > 0$ and
  $pot(p) \le m(P_{i - 1})$, there exists $m' \in \cR_i(m)$ such that:
  \begin{equation}\label{eq:movingtokens}
    m'(p) = m(p) - 1\,,~ m'(q) = m(q) + 1\,,~ \forall r \in
    P_i \setminus \{p, q\}, \ m'(r) = m(r)\,.
  \end{equation}
  A marking is live if and only if it satisfies the
  $n$-condition.
\end{prop}

\begin{proof}
  Consider an $i$-live marking $m$.
  For any $j \le i$, there is a marking $m' \in \cR_i(m)$ which enables a
  transition of $T_j$. This marking satisfies $\varphi_j(m') \le m'(P_{j - 1})$.
  By (weak) reversibility, $m \in \cR(m')$, so $\varphi_j(m) \le m(P_{j - 1})$
  (Lemma \ref{lem:preserve}). Since the number of tokens in $P_i$ is the same
  for all the markings of $\cR_i(m)$, $m(P_i) > 0$ (otherwise, the transitions
  of $T_i$ would be dead).

  \medskip

  \noindent  We prove the reverse direction and the second part of the proposition
  by induction on $i \ge 1$, i.e. :\\
  If $m$ satisfies the $i$-condition then:\\ 
  (1)
  for every $p, q \in P_i$ such that $p \neq q$, $m(p) > 0$ and
  $pot(p) \le m(P_{i - 1})$, there exists $m' \in \cR_i(m)$ such that:
  $
  m'(p) = m(p) - 1\,,~ m'(q) = m(q) + 1\,,~ \forall r \in
  P_i \setminus \{p, q\}, \ m'(r) = m(r)\,.
  $\\ 
  (2) $m$ is $i$-live.

  \smallskip

  \noindent  The case $i = 1$ is trivial. 

  \noindent  Suppose that the claim has been proven for all $j \le i - 1$. Let $m$ be a
  marking which satisfies the $i$-condition. Consider two cases: $pot(p) = 0$
  and $pot(p) > 0$.

  \noindent  If $pot(p) = 0$ then
  the output transitions of $p$ are enabled by $m$.
  For any arbitrary $q\neq p$,
  fire the transitions along a path from $p$ to
  $q$ in $T_i$, we obtain a marking $m'$ satisfying \eref{eq:movingtokens}.
  So we have proved assertion (1). 
  Now choose some $q$ such that $pot(q) > 0$ 
  (there is at least one). Then $m'(P_{i - 1}) \ge pot(q) > 0$. By the
  induction hypothesis, $m'$ is $(i - 1)$-live. Moreover, $m'$ enables the
  output transitions of $q$. Hence $m'$ is $i$-live, which implies 
  $m$ is $i$-live.

  \noindent  If $pot(p) > 0$ then $m(P_{i - 1}) > 0$, hence $m$ is $(i - 1)$-live by the
  induction hypothesis. It remains to find a marking in $\cR_{i - 1}(m)$ which
  enables the output transitions of $p$. If for all $r \in P_{i - 1}$,
  $m(r) \ge pot(p, r)$ then choose $m$. Otherwise, choose a marked place
  $q$ of $P_{i - 1}$ such that $pot(q) \le m(P_{i - 2})$
  and a level $(i - 1)$ interface place $q'$, then apply the
  induction hypothesis on \eref{eq:movingtokens}
  to find $m_1 \in \cR_{i - 1}(m)$ such that $m_1(q) =
  m(q) - 1$, $m_1(q') = m(q') + 1$ and $m_1(r) = m(r)$ for every other places
  $r$ of $P_{i - 1}$. We have $\varphi_{i - 1}(m_1) = \mbox{max}\{pot(r), r \in
  P_{i - 1}\}$.
  Now starting from $m_1$, repeat the
  following procedure:

  \smallskip

  \begin{itemize}[nolistsep]
    \item Step 1: Find two place $r_1$, $r_2$ in $P_{i - 1}$ such that
      $\bar m(r_1) < pot(p, r_1)$  and $\bar m(r_2) > pot(p, r_2)$, 
      $\bar m$ denoting the current marking.
    \item Step 2: Use the induction hypothesis on \eref{eq:movingtokens} to find
      $\bar m' \in \cR_{i - 1}(\bar m)$ such that
      $\bar m'(r_1) = \bar m(r_1) + 1$, $\bar m'(r_2) = \bar m(r_2) -
      1$ and $\bar m'(r) = \bar m(r)$ for all  $r \in P_{i - 1}
      \setminus \{r_1, r_2\}$.
  \end{itemize}

  \smallskip

  \noindent  All the intermediate markings are $(i - 1)$-live. Since
  $\bar m(P_{i - 1}) \ge pot(p)$, if there exists $r_1 \in
  P_{i - 2}$ such that $\bar m(r_1) < pot(p, r_1)$ then there exists
  $r_2 \in P_{i - 2}$ such that $\bar m(r_2) > pot(p, r_2) \ge 0$ as well. 
  Because the interface places
  have maximal potential, at the beginning of each iteration, we always have
  $\varphi_{i - 1}(\bar m) = \mbox{max}\{pot(r), r \in P_{i - 1}\}$, hence
  $pot(r_2) \le \varphi_{i - 1}(\bar m) \le \bar m(P_{i - 2})$. 
  Each iteration strictly diminishes 
  the number of ``missing'' tokens in the places of
  $P_{i - 1}$ synchronised with $p$, so the procedure eventually stops at a
  marking $m_2$ such that $m_2(r) \ge pot(p, r)$ for every place $r \in
  P_{i - 1}$. This marking enables the output transitions of $p$.
\end{proof}

\noindent{\bf Example:} The ordered \piun-net in Figure \ref{fig:pi3} is a
\pitrois-net. Consider two markings: $m_1 = p_3 + q_3 + r_3 + 4 p_1$ and
$m_2 = 3 q_3 + 4 p_1$. These markings agree on all the S-invariants, but
only $m_1$ satisfies the 3-condition. It is easy to check that 
$m_1$ is live while $m_2$ is dead.

We conclude this subsection by showing that the reachability problem for
\pitrois-nets can be efficiently decided as well.

\begin{theo}
  \label{thm:reachability}
  Suppose that the initial marking $m_0$ is live. Then the reachability set
  $\cR(m_0)$ coincides with the set $\cS(m_0)$ of markings which satisfy the
  $n$-condition and agree with $m_0$ on the S-invariants given by Corollary
  \ref{cor:sinvariants}.
\end{theo}

\begin{proof}
  The inclusion $\cR(m_0) \subset \cS(m_0)$ is the combination of the results 
  of Corollary
  \ref{cor:sinvariants} and Proposition \ref{prop:ilivecondition}. 

  \noindent  To prove the converse, we look for a marking which is 
  reachable from every marking of $\cS(m_0)$. Let $p_j$, $1 \le j \le n$, be a
  place of maximal potential of $P_j$, that is,
  $pot(p_j) = \mbox{max}\{pot(p), p \in P_j\}$. Let $m_0'$ denote the unique
  marking in $\cS(m_0)$ such that $m_0'(p) = 0$ for every $ p \notin \{p_1,
  \dots, p_n\}$.
  Consider an arbitrary marking $m$ in $\cS(m_0)$.
  We prove by a reverse induction on $i\leq n$ and by using the second part of 
  Proposition \ref{prop:ilivecondition} that
  there exists a marking $m' \in \cR(m)$ such that
  $m'(p) = 0~\forall p \notin \{p_1, \dots, p_n\}$.
  The inductive claim is:
  \vspace*{-0.3cm}
  \begin{center}
    There exists a marking $m'_i \in \cR(m)$ such that
    $\forall p \in \cup_{i\leq j \leq n} P_j \setminus \{p_i, \dots, p_n\}\ m'_i(p) = 0$\\
    and $m'_i$ satisfies the $i-1$ condition.
  \end{center}  

  \vspace*{-0.3cm}
  \noindent
  Let us address the basis case $i=n$. Assume that there exists $p\neq p_n$ such that
  $m_0(p)>0$. Using proposition~\ref{prop:ilivecondition}, we move a token
  from $p$ to $p_n$. 
  Furthermore by lemma~\ref{lem:preserve}, the $n$-condition
  is still satisfied.
  Iterating this process, we obtain a marking $m'_n$
  such that $\forall p \in P_n \setminus \{p_n\}\ m'_n(p) = 0$
  and the $n$-condition is still satisfied.
  The inductive case is similar by observing that the sequence
  that moves the tokens of $P_i$ does not use transitions
  of $T_j$ for $j>i$.

  \noindent
  Since $m'$ is also an
  element of $\cS(m_0)$, $m' = m_0'$. So $m_0'$ is reachable from every marking
  in $\cS(m_0)$. By (weak) reversibility, every marking in $\cS(m_0)$ is reachable
  from $m_0'$. So $\cS(m_0) \subset \cR(m_0') = \cR(m_0)$.
\end{proof}

\subsection{Computing the normalising constant}
\label{subsec:compute}
The normalising constant of a product-form Petri net (see Section
\ref{subsec:spn}) is
$G = \sum_m \mathbb 1_{m \in \cR(m_0)} \prod_{p \in P} u_p^{m(p)}$.
It is in general a difficult task to compute $G$, as can be guessed
from the complexity of the reachability problem. 
However, efficient algorithms may exist for nets with a well-structured
reachability set. Such algorithms were known for Jackson
networks~\cite{reis80} and the {\em S-invariant reachable} Petri nets defined 
in~\cite{coheta}. We show that is is also the case for the class
of live \pitrois-nets which is strictly larger than the class of
Jackson networks (which correspond to 1-level ordered nets) and is not included
in the class of S-invariant reachable Petri nets.



\smallskip

Suppose that $m_0$ is a live marking.
Suppose that the places of each level are ordered by increasing potential:
$P_i = \{p_{i1},
\dots, p_{ik_i}\}$ such that $\forall 1 \le j < k_i$, $pot(p_{ij}) \le pot(p_{i(j + 1)})$.

Let $V$ denote the $n \times P$-matrix the $i$-th row of which is the S-invariant
$v_i$ defined in Corollary \ref{cor:sinvariants}.
For $1 \le i \le n$, set $C_i = v_i m_0 = \widetilde m_0 (P_i)$.
Then the reachability set consists of all $n$-live markings $m$ such that
$Vm = {}^t(C_1, \dots, C_n)$. 

\smallskip

For $1 \le i \le n$, $1 \le j \le k_i$
and $c_1, \dots, c_i \in \mathbb Z$, define $E(i, j, c_1, \dots, c_i)$ as
the set of markings $m$ such that 
\[
\begin{cases}
  m(p_{i\nu}) = 0 \mbox{ for all } \nu > j\\
  Vm = {}^t(c_1, \dots, c_i, 0 \dots, 0)\\
  \varphi_\nu(m) \le m(P_{\nu - 1}) \mbox{ for all } 2 \le \nu \le i\,.
\end{cases}
\]
The elements of $E(i, j, c_1, \dots, c_i)$ are the markings which satisfy the
second part of the $i$-condition and the S-invariants constraints 
$(c_1, \dots, c_i, 0, \dots, 0)$ and
concentrate tokens in $P_1, \dots, P_{i - 1}$ and $\{p_{i1}, \dots,
p_{ij}\}$.

With each $E(i, j, c_1, \dots, c_i)$ associate\\ 
\centerline{
$G(i, j, c_1, \dots, c_i) = \pi(E(i, j, c_1, \dots, c_i)) = \sum
\prod_{p \in P}u_p^{m(p)}$}\\
the sum being taken over all $m \in E(i, j, c_1, \dots, c_i)$.

\smallskip

We propose to compute $G(n, k_n, C_1, \dots, C_n)$ by dynamic programming.
It consists in breaking each $G(i,
j, c_1, \dots, c_i)$ into smaller sums.
This corresponds to a partition of the elements of 
$E(i, j, c_1, \dots, c_i)$ by the number of tokens in 
$p_{ij}$.

\begin{prop}
  \label{prop:partitionE}
  Let be given $E = E(i, j, c_1, \dots, c_i)$. If $c_i < 0$ then $E =
  \emptyset$. If $c_i \ge 0$ then for every non-negative integer $a$:

  \smallskip

  \begin{enumerate}[nolistsep]
    \item If $a > c_i$ then $E \cap \{m | m(p_{ij}) = a\} = \emptyset$.
    \item If $a < c_i$ and $j = 1$ then $E \cap \{m | m(p_{ij}) = a\} =
      \emptyset$.
    \item If $a < c_i$ and $j \ge 2$ then\\ 
    $E \cap \{m | m(p_{ij}) = a\} = \{m+ap_{ij} \mid m \in  E(i, j - 1, c_1 -
    v_1(a p_{ij}), \dots, c_i -	v_i(a~p_{ij}))\}$.
    \item If $a = c_i$ and $i = 1$ then $E \cap \{m | m(p_{ij}) = a\} =
      \{c_1 p_{1j}\}$.
    \item If $a = c_i$ and $i > 1$ then\\ 
    $E \cap \{m | m(p_{ij}) = a\} = \{~m+ap_{ij} \mid m \in E(i - 1, k_{i - 1}, c_1 - v_1(a p_{ij}), \dots, c_{i - 1} - v_{i -
      1}(a p_{ij}))\}$.
  \end{enumerate}
\end{prop}

\begin{proof}
  Suppose that $E \neq \emptyset$. Let $m$ be an element of $E$ such that
  $m(p_{ij}) = a$. We have $m(P_i) = c_i$, so $a \le c_i$. Moreover, if
  $m(p_{ij}) < m(P_i)$ then $m$ must mark some place $p_{i\nu}$ with $\nu < j$,
  so $j \ge 2$. These prove the first and the second cases.

  \medskip

\noindent  The fourth case is trivial.

  \medskip

\noindent  Let us address the third case, we have to show that:
  \begin{align}
    &\forall m \in E \mbox{ s.t. } m(p_{ij}) = a, (m - a p_{ij}) \in  E(i, j - 1, c_1 - v_1(a p_{ij}), \dots, c_i -
	v_i(a p_{ij})) \label{eq:etoe'1}\\
    &\forall m' \in E(i, j - 1, c_1 - v_1(a p_{ij}), \dots, c_i -
	v_i(a p_{ij})),(m' + a p_{ij}) \in E \label{eq:e'toe1}
  \end{align}

\noindent
The values $c_1 - v_1(a p_{ij}), \dots, c_i - v_i(a p_{ij})$ are
  obtained by:
  \begin{align*}
    Vm &= {}^t(c_1, \dots, c_i, 0, \dots, 0) \\
    \iff V(m - a p_{ij}) &=  {}^t(c_1 - v_1(a p_{ij}), \dots, c_i - v_i(a p_{ij}),
    0, \dots, 0)\,.
  \end{align*}

\noindent   We have to show that 
  $\varphi_\nu(m - a p_{ij}) \le (m - a
  p_{ij})(P_{\nu - 1})~\forall 2 \le \nu \le i$ and $\varphi_\nu(m' + a p_{ij})
  \le (m' + a p_{ij})(P_{\nu - 1})~\forall 2 \le \nu \le i$.

\noindent  Since $m$ and $(m - a p_{ij})$ only differ at $p_{ij}$, 
  it suffices to show that
  $\varphi_i(m - a p_{ij}) \le (m - a p_{ij})(P_{i - 1})$. Indeed, $\varphi_i(m
  - a p_{ij}) = \varphi_i(m)$ because both markings mark some $p_{i\nu}$ with
  $\nu < j$, and $(m - a p_{ij})(P_{i - 1}) = m(P_{i - 1})$ because the two
  markings are identical on $P_{i - 1}$.

\noindent  Similarly, given $m' \in E(i, j - 1, c_1 - v_1(a p_{ij}), \dots, c_i -
	v_i(a p_{ij}))$, to prove \eref{eq:e'toe1}, it suffices to show
  that $\varphi_i(m' + a p_{ij}) \le (m' + a p_{ij})(P_{i - 1})$. 
  Indeed, 
  $(m' + a p_{ij})(P_{i - 1}) = m'(P_{i - 1}) \le
  \varphi_i(m') \le \varphi_i(m' + a p_{ij})$.

  \medskip

\noindent  The fifth case is similar. It suffices to show that
  $\varphi_{i - 1}(m - a p_{ij}) \le (m - a p_{ij})(P_{i - 2})$ and
  $\varphi_i(m' + a p_{ij}) \le (m' + a p_{ij})(P_{i - 1})$.
  The first inequality is immediate since $(m - a p_{ij})$ is the restriction
  of $m$ on $\bigcup_{1 \le \nu \le i - 1} P_\nu$.
  To prove the second one, note that $(m' + a
  p_{ij})(P_{i - 1}) = m'(P_{i - 1}) = c_{i - 1} - v_{i - 1}(a p_{ij}) =
  m(P_{i - 1})$ and $\varphi_i(m' + a p_{ij}) = \varphi_i(m)$.
\end{proof}

The proposition \ref{prop:partitionE} induces the following relations between the
sums $G(i, j, c_1, \dots, c_i)$.

\begin{coro}\label{cor:G}
  If $c_i < 0$ then $G(i, j, c_1, \dots, c_i) = 0$. If $c_i \ge 0$ then:
  \begin{small}
  \begin{itemize}[nolistsep]
    \item Case $2 \le i \le n$, $2 \le j \le k_i$:
      \begin{align*}
	G(i, j, c_1, \dots, c_i) = & \sum_{\nu = 0}^{c_i - 1}u_{p_{ij}}^\nu G(i, j - 1,
	c_1 - v_1(\nu p_{ij}), \dots, c_i - v_i(\nu p_{ij})) \\
	& + u_{p_{ij}}^{c_i} G(i - 1, k_{i - 1}, c_1 -
	v_1(c_i p_{ij}), \dots, c_{i - 1} - v_{i - 1}(c_i p_{ij}))\,.
      \end{align*}
    \item Case $2 \le i \le n$, $j = 1$:
      \begin{align*}
	G(i, 1, c_1, \dots, c_i) = u_{p_{i1}}^{c_i} G(i - 1, k_{i - 1}, c_1 -
	v_1(c_i p_{i1}), \dots, c_{i - 1} - v_{i - 1}(c_i p_{i1}))\,.
      \end{align*}
    \item Case $i = 1$, $j \ge 2$:
	$G(1, j, c_1) = \sum_{\nu = 0}^{c_1 - 1} u_{p_{1j}}^\nu G(1, j - 1, c_1 -
	\nu) + u_{p_{1j}}^{c_1}\,.$
    \item Case $i = 1$, $j = 1$:
	$G(1, 1, c_1) = u_{p_{11}}^{c_1}\,.$
  \end{itemize}
  \end{small}
\end{coro}

\noindent{\bf Complexity.}
Since $i \le n$, $j \le K = \mbox{max}\{k_1, \dots, k_n\}$, the number of
evaluations is bounded by $n \times K \times \gamma$, where $\gamma$ upper
bounds the $c_i$'s.
Let $\alpha$ denote the global maximal potential. From
\eref{eq:effectivemarking}, we obtain $\gamma = \mathcal O(m_0(P) K^n
\alpha^n)$. So the complexity of a dynamic programming algorithm using Cor.
\ref{cor:G} is $\mathcal O(m_0(P) n K^{n + 1} \alpha^n)$, i.e.
pseudo-polynomial for a fixed number of state machines.


\section{Perspectives}
\label{section:conclusion}


This work has several perspectives. First, we are interested 
in extending and applying our rules for a modular modelling of complex product-form
Petri nets. We also want to obtain characterisation of product-form
Petri nets when stochastic Petri nets are equipped with infinite-server policy.
Then we want to validate the formalism of  $\pitrois$-nets showing
that it allows to express standard patterns of distributed systems.
We plan to implement analysis of  $\pitrois$-nets and
integrate it into a tool for stochastic Petri nets like GreatSPN~\cite{GreatSPN}.
Finally we conjecture that reachability is {\sf EXPSPACE}-complete for $\pideux$-nets
and we want to establish it.

\closegraphsfile

\bigskip

\noindent
{\bf Acknowledgements.} We would like to thank the anonymous referees
whose numerous and pertinent suggestions have been helpful in
preparing the final version of the paper.


\begin{thebibliography}{10}


\bibitem{gspnbook}
M.~Ajmone~Marsan, G.~Balbo, G.~Conte, S.~Donatelli, G.~Franceschinis.
\newblock {\em Modelling with Generalized Stochastic Petri Nets}.
\newblock John Wiley \& Sons, 1995.


\bibitem{BCOQ}
F.~Baccelli, G.~Cohen, G.J. Olsder, and J.P. Quadrat.
\newblock {\em Synchronization and Linearity}.
\newblock John Wiley \& Sons, New York, 1992.

\bibitem{Balsamo2011}
S. Balsamo, P. G. Harrison and A. Marin.
\newblock Methodological construction of product-form stochastic Petri nets for performance evaluation
\newblock {\em Journal of Systems and Software},
available on line, to appear.

\bibitem{BalsamoMarin2011}
S. Balsamo and A. Marin.
\newblock Performance engineering with product-form models: efficient solutions and applications
\newblock {\em Proceedings of the second joint WOSP/SIPEW international conference on Performance engineering (ICPE'11)},
ACM publisher, pp. 437-448, Karlsruhe, Germany, 2011.


\bibitem{BM2009}
S. Balsamo, A. Marin.
\newblock Composition of product-form Generalized Stochastic Petri Nets: a modular approach.
\newblock {\em Proc. ESM 2009, Eurosis 23rd European Simulation and Modelling Conference}, 
United Kingdom, Leicester, October 2009.

\bibitem{bcmp}
F.~Baskett, K.~M. Chandy, R.~R. Muntz, F.~Palacios.
\newblock Open, closed and mixed networks of queues with different classes of
  customers.
\newblock {\em Journal of the ACM}, 22(2):248--260, April 1975.

\bibitem{bouchserercim}
R.~J. Boucherie, M.~Sereno.
\newblock On closed support {T}-invariants and traffic equations.
\newblock {\em Journal of Applied Probability}, (35): 473--481, 1998.

\bibitem{GreatSPN}
G. Chiola, G. Franceschinis, R. Gaeta, M. Ribaudo. 
\newblock GreatSPN 1.7: Graphical Editor and Analyzer for Timed and Stochastic Petri Nets. 
\newblock {\em Performance Evaluation} 24(1-2): 47-68 (1995)

  

\bibitem{DeEs95}
  J.~Desel and J.~Esparza
  \newblock {\em Free Choice {P}etri Nets}, volume~40 of {\em Cambridge Tracts
  Theoret. Comput. Sci.}
  \newblock Cambridge Univ. Press, 1995.


\bibitem{survey94}
J. Esparza and M. Nielsen. 
\newblock Decidability issues for Petri nets - a survey. 
\newblock {\em Journal of Informatik Processing and Cybernetics}, 30(3):143-160, 1994. 

\bibitem{coheta}
J.L. Coleman, W.~Henderson, P.G. Taylor.
\newblock Product form equilibrium distributions and a convolution algorithm
  for stochastic {Petri} nets.
\newblock {\em Performance Evaluation}, 26(3):159--180, September 1996.

\bibitem{esparza94}
J. Esparza.
\newblock Reduction and Synthesis of Live and Bounded Free Choice Petri Nets.
\newblock {\em Information and Computation},  114(1):50--87, 1994

\bibitem{Feinberg79}
M. Feinberg. 
\newblock Lectures on chemical reaction networks. 
\newblock {\em At the Math. Research Center, Univ. Wisconsin, 1979}.
Available online at:\\ \begin{footnotesize}{\tt
  http://www.che.eng.ohio-state.edu/$\sim$feinberg/LecturesOnReactionNetworks}.\end{footnotesize}


\bibitem{Haddad05}
S.~Haddad, P.~Moreaux, M.~Sereno, M.~Silva
\newblock Product-form and stochastic Petri nets: a structural approach. 
\newblock {\em Performance Evaluation}, 59: 313-336, 2005.


\bibitem{HMNg11}
  S.~Haddad, J.~Mairesse, H.-T.~Nguyen
  \newblock Synthesis and analysis of product-form Petri nets.
  \newblock {\em Petri Nets 2011}, LNCS 6709, 288-307, 2011.

\bibitem{Harrison2011}
P. G. Harrison and L. M. Catalina
\newblock Hierarchically constructed Petri-nets and product-forms.
\newblock {\em Proceedings of the 5th International ICST Conference on Performance Evaluation Methodologies and Tools (VALUETOOLS'11)},
ICST, pp. 101-110, Paris, France, 2011.


\bibitem{heluta-jams}
W.~Henderson, D.~Lucic, P.G. Taylor.
\newblock A net level performance analysis of stochastic {Petri} nets.
\newblock {\em Journal of Australian Mathematical Soc. Ser. B}, 31:176--187, 1989.

\bibitem{jack63}
J. R.~Jackson.
\newblock Jobshop-like Queueing Systems.
\newblock {\em Management Science}, 10(1): 131--142, 1963.
  


\bibitem{kell79}
F.~Kelly.
\newblock {\em Reversibility and Stochastic Networks}.
\newblock Wiley, New-York, 1979.

\bibitem{lazar-robert87}
A.~A.~Lazar,   T.~G.Robertazzi.
\newblock {Markovian} {Petri} Net Protocols with Product Form Solution.
\newblock {\em Proc. of INFOCOM 87},
pp. 1054--1062, San Francisco, CA, USA, 1987.


\bibitem{li_geor91}
M.~Li, N.~D.~Georganas.
\newblock Parametric Analysis of Stochastic {Petri} Nets,
\newblock {\em Fifth International Conference on Modelling and Tools for Computer Performance Evaluation},
Torino, Italy, 1991,

\bibitem{MairesseNguyen09}
J.~Mairesse, H-T.~Nguyen.
\newblock  Deficiency Zero Petri Nets and Product Form. 
\newblock {\em Petri Nets 2009}, LNCS 5606, 103-122, 2009.

\bibitem{MM82}
E.~Mayr, A.~Meyer. 
\newblock The complexity of the word problem for commutative semigroups an polynomial ideals.
\newblock {\em  Advances in Math}, 46 (1982), 305-329.


\bibitem{Papadimitriou94}
C.~Papadimitriou.
\newblock Computational Complexity.
\newblock {\em Addison Wesley}, 1994.

\bibitem{reis80}
M.~Reiser,  S.S.~Lavenberg.
\newblock Mean Value Analysis of Closed Multichain Queueing Networks.
\newblock {\em Journal of the ACM}, 27(2): 313-322, 1980.

\end{thebibliography}
\end{document}